\documentclass[journal]{IEEEtrans}
\usepackage{amsmath}
\usepackage{amsfonts}
\usepackage{float}
\usepackage{paralist}
\usepackage{graphicx}
\usepackage{amsthm}
\usepackage{amssymb}
\usepackage{multirow}

\usepackage{subcaption}
\usepackage{graphics,graphicx} 
\usepackage[pagebackref=false,breaklinks=true,colorlinks=true, citecolor=magenta,linkcolor=cyan,urlcolor=blue,bookmarks=true]{hyperref}
\usepackage[ruled,linesnumbered]{algorithm2e}
\SetKwInput{KwData}{Input}
\usepackage{accents}
\usepackage{color}
\usepackage{algorithmicx}
\usepackage[noend]{algpseudocode}
\def\BState{\State\hskip-\ALG@thistlm}
\usepackage{supertabular,booktabs}
\usepackage{multicol,lipsum}
\usepackage{stfloats}
\usepackage{arydshln}
\usepackage{mathtools}
\usepackage{graphicx}
\usepackage{tikz}
\usetikzlibrary{shapes,arrows}
\usetikzlibrary{positioning}	
\usepackage{chngcntr}
\usepackage{apptools}
\usepackage{comment}
\newcommand{\RNum}[1]{\uppercase\expandafter{\romannumeral #1\relax}}
\newcommand\norm[1]{\lVert#1\rVert}

\graphicspath{ {images/} }

\newtheoremstyle{theoremdd}
{\topsep}
{\topsep}
{\itshape}
{0pt}
{\bfseries}
{:}
{ }
{\thmname{#1}\thmnumber{ #2} \thmnote{(#3)}}

\theoremstyle{theoremdd}

\AtAppendix{\counterwithin{assumption}{section}}
\AtAppendix{\counterwithin{definition}{section}}
\AtAppendix{\counterwithin{theorem}{section}}

 \usepackage[compact]{titlesec}
 \titlespacing{\section}{0pt}{12pt plus 4pt minus 4pt}{0pt plus 2pt minus 2pt}
 \titlespacing{\subsection}{0 pt}{12pt plus 4pt minus 0pt}{0pt plus 2pt minus 2pt}
\def\BibTeX{{\rm B\kern-.05em{\sc i\kern-.025em b}\kern-.08em
		T\kern-.1667em\lower.7ex\hbox{E}\kern-.125emX}}
\begin{document}
	\title{Robust Policy Optimization in  Continuous-time Mixed $\mathcal{H}_2/\mathcal{H}_\infty$ 
		Stochastic Control}
	
	\author{Leilei Cui and 
		Lekan Molu (\IEEEmembership{Member, IEEE}) %
		\thanks{L. Cui is with the Control and Networks Lab, Department of Electrical and Computer Engineering, Tandon School of Engineering, New York University, Brooklyn, NY 11201, USA. (email: l.cui@nyu.edu).}
		\thanks{L. Molu is with the Reinforcement Learning group at Microsoft Research, 300 Lafayette Street, New York, NY 10012, USA. (email: lekanmolu@microsoft.com).}%
		
	}
	
	\maketitle
	\IEEEpeerreviewmaketitle
	\begin{abstract}
Following the recent resurgence in establishing linear control theoretic benchmarks for reinforcement leaning (RL)-based policy optimization (PO) for complex dynamical systems with continuous state and action spaces, 
an optimal control problem for a continuous-time infinite-dimensional linear stochastic system possessing additive Brownian motion is optimized on a cost that is an exponent of the quadratic form of the state, input, and disturbance terms. We lay out a model-based and model-free algorithm for RL-based stochastic PO. For the model-based algorithm, we establish rigorous convergence guarantees. For the sampling-based algorithm, over trajectory arcs that emanate from the phase space, we find that the Hamilton-Jacobi Bellman equation parameterizes trajectory costs --- resulting in a discrete-time (input and state-based) sampling scheme accompanied by unknown nonlinear dynamics with continuous-time policy iterates. The need for known dynamics operators is circumvented and we arrive at a reinforced PO algorithm (via policy iteration) where an upper bound on the  $\mathcal{H}_2$ norm is minimized (to guarantee stability) and a robustness metric is enforced by maximizing the cost with respect to a controller that includes the level of noise attenuation specified by the system's $H_\infty$ norm. Rigorous robustness analyses is prescribed in an input-to-state stability formalism. 
Our analyses and contributions are distinguished by many natural systems characterized by  additive Wiener process, amenable to \^Ito's stochastic differential calculus in dynamic game settings. 
\end{abstract}

\begin{IEEEkeywords}
Optimal control, Robust control, $\mathcal{H}_\infty$ control, Iterative learning control, Machine learning.
\end{IEEEkeywords}

	\definecolor{light-blue}{rgb}{0.30,0.35,1}
\definecolor{light-green}{rgb}{0.20,0.49,.85}
\definecolor{purple}{rgb}{0.70,0.69,.2}

\newcommand{\lb}[1]{\textcolor{light-blue}{#1}}
\newcommand{\bl}[1]{\textcolor{blue}{#1}}

\newcommand{\maybe}[1]{\textcolor{gray}{\textbf{MAYBE: }{#1}}}
\newcommand{\inspect}[1]{\textcolor{cyan}{\textbf{CHECK THIS: }{#1}}}
\newcommand{\more}[1]{\textcolor{red}{\textbf{MORE: }{#1}}}

\renewcommand{\figureautorefname}{Fig.}
\renewcommand{\sectionautorefname}{$\S$}
\renewcommand{\equationautorefname}{equation}
\renewcommand{\subsectionautorefname}{$\S$}
\renewcommand{\chapterautorefname}{Chapter}

\newcommand{\cmt}[1]{{\footnotesize\textcolor{red}{#1}}}
\newcommand{\todo}[1]{\textcolor{cyan}{TO-DO: #1}}
\newcommand{\review}[1]{\noindent\textcolor{red}{$\rightarrow$ #1}}
\newcommand{\response}[1]{\noindent{#1}}
\newcommand{\stopped}[1]{\color{red}STOPPED HERE #1\hrulefill}

\newcounter{mnote}
\newcommand{\marginote}[1]{\addtocounter{mnote}{1}\marginpar{\themnote. \scriptsize #1}}
\setcounter{mnote}{0}
\newcommand{\ie}{i.e.\ }
\newcommand{\eg}{e.g.\ }
\newcommand{\cf}{cf.\ }
\newcommand{\yes}{\checkmark}
\newcommand{\no}{\ding{55}}

\newcommand{\flabel}[1]{\label{fig:#1}}
\newcommand{\seclabel}[1]{\label{sec:#1}}
\newcommand{\tlabel}[1]{\label{tab:#1}}
\newcommand{\elabel}[1]{\label{eq:#1}}
\newcommand{\alabel}[1]{\label{alg:#1}}
\newcommand{\fref}[1]{\cref{fig:#1}}
\newcommand{\sref}[1]{\cref{sec:#1}}
\newcommand{\tref}[1]{\cref{tab:#1}}
\newcommand{\eref}[1]{\cref{eq:#1}}
\newcommand{\aref}[1]{\cref{alg:#1}}

\newcommand{\bull}[1]{$\bullet$ #1}
\newcommand{\argmax}{\text{argmax}}
\newcommand{\argmin}{\text{argmin}}
\newcommand{\mc}[1]{\mathcal{#1}}
\newcommand{\bb}[1]{\mathbb{#1}}

\def\kau{\mc{K}}
\def\particle{\bm{x}}
\def\materialresponse{\bm{G}}
\def\orthoggroup{{\textit{SO}}(3)}
\def\liegroup{{\textit{SE}}(3)}
\def\liealgebra{\mathfrak{se}(3)}
\def\identity{\bm{I}}
\def\sgmin{\underline{\sigma}}
\def\sgmax{\bar{\sigma}}
\def\eigmin{\underline{\lambda}}
\def\eigmax{\bar{\lambda}}
\def\df{\mathrm{d}}
\newcommand{\trace}[1]{\textbf{Tr}(#1)}

\def\rot{{R}}
\def\rthree{\bb{R}^3}
\def\reline{\bb{R}}
\def\ren{\bb{R}^n}
\def\skew{S}
\def\jacob{J}
\def\state{\bm{x}}
\def\statex{x}
\def\statey{y}
\def\statez{z}
\def\hot{h.o.t.\ }
\def\lhs{l.h.s.\ }
\def\rhs{r.h.s.\ }
\def\hinf{\mc{H}_\infty}
\def\htwo{\mc{H}_2}
\def\identity{I}
\def\costdiff{\mathbf{\tilde{V}}}
\def\gain{\bm{k}}

\newtheorem{assumption}{Assumption}
\newtheorem{problem}{Problem}
\newtheorem{proposition}{Proposition}
\newtheorem{theorem}{Theorem}
\newtheorem{lemma}{Lemma}
\newtheorem{remark}{Remark}
\newtheorem{corollary}{Corollary}

\def\vec{\texttt{vec}}
\def\vecs{\texttt{vecs}}
\def\vect{\texttt{vect}}
\def\svec{\texttt{svec}}
\def\vecv{\texttt{vecv}}
\def\mat{\texttt{mat}}
\def\smat{\texttt{smat}}
\def\Tr{Tr}

	\section{Introduction}
Lately, various system-theoretic results analyzing the global convergence~\cite{Fazel2018} and computational complexity~\cite{Mohammadi2022} of nonconvex, constrained~\cite{HuAnnualRevs} gradient-based~\cite{Gravell2021} and derivative-free~\cite{zhang2021derivative} policy optimization in sampling-based reinforcement learning (RL) when the complete set of decision (or state feedback) variables are not previously known have appeared as control benchmarks~\cite{Zhang2020SIAM, Zhang2021}. The most basic setting consists in optimizing over a decision variable $K$ which must be determined from a (restricted) class of controllers $\mc{K}$ \ie $	\min_{K \in \mc{K}} J(K)$
%
%
where $J(K)$ is an objective (e.g. tracking error, safety assurance, goal-reaching measure of performance e.t.c.) required to be satisfied.  In principle, $K$ can be realized as a linear controller, a linear-in-the-parameters polynomial, or as a nonlinear kernel in the form of a radial basis function, or neural network.

These policy optimization (PO) schemes apply to a broad range of problems and have enjoyed wide success in complex systems where analytic models are difficult to derive~\cite{LevineEnd2End}. While they have become a popular tool for modern learning-based control~\cite{RechtTour}, the theoretical underpinning of their convergence, sample complexity, and robustness guarantees are little understood \textit{in the large}. Only recently have  
rigorous analyses tools emerged~\cite{Zhang2020SIAM, Pang2021} for benchmarking RL with deterministic and additive Gaussian disturbance linear quadratic (LQ) controllers~\cite{DeanSampleComplexity, Fazel2018}. 

Tools for analyzing the convergence, sample complexity, or robustness of RL-based PO largely fall into one of infinite-horizon 
\begin{inparaenum}[(i)]
	\item discrete-time LQ regulator (LQR) settings \ie
	$$ \min_{K\in \mc{K}} \bb{E} \sum_{t=0}^\infty  (x_t^\top Q x_t + u_t^\top R u_t) \text{ s.t. } x_{t+1} = A x_t + B u_t, x_0 \sim \mathcal{P}_0$$ where $A, B, Q, \text{ and } R$ are standard LQR matrices for state $x_t$, control input $u_t$ and $x_0$ is drawn from a random distribution $P_0$
	~\cite{Fazel2018};
	\item discrete-time  LQ problems under multiplicative noise \ie $\min_{\pi \in \Pi} \mathbb{E}_{x_0, \{\delta_i\}, \{\gamma_i\}\}} \sum_{t=0}^{\infty} (x_t^\top Q x_t + u_t^\top R u_t)$ $\text{subject to } x_{t+1} = (A + \sum_{i=1}^{p}\delta_{ti} A_i) x_t + (B + \sum_{i=1}^{q}\gamma_{ti} B_i) u_t$ with covariance $\bb{E}_{x_0}[x_0x_0^T]$ 
	and $A, B, Q, R$ are the standard LQR matrices with $\delta_{ti}$ and $\gamma_{tj}$ serving as the i.i.d zero-mean and mutually independent multiplicative noise terms
	~\cite{Gravell2021}; or 
	\item Risk-sensitive $\hinf$-control~\cite{Glover1989} and discrete- and continuous-time  mixed $\htwo/\hinf$ design~\cite{Khargonekar1988, HuAnnualRevs} where the upper bound on the $\htwo$ cost is minimized subject to satisfying a  set of risk-sensitive (often $\hinf$) constraints that \textit{attenuate}~\cite{basar1990minimax} an unknown disturbance. \ie $\min_{K \in \mc{K}} J(K):=Tr(P_K DD^\top)$ $\text{subject to } \mc{K} := \{K | \rho(A-BK) < 1, \|T_{zw}(K)\|_\infty < \gamma \}$ where $P_K$ is the solution to the generalized algebraic Riccati equation (GARE), $A,B, D, K$ are standard closed-loop system matrices, $\|T_{zw}(K)\|_\infty$ denotes the $\hinf$-norm of the closed-loop transfer function from a disturbance input $w$ to its output $z$, and $\gamma>0$, Here, $\gamma>0$, upper-bounded by $\gamma^\star$, a scalar measure of system risk-sensitivity~\cite{book_Basar}.
\end{inparaenum}  

\textit{We focus on continuous-time linear systems in which disturbances enter additively as random stochastic Wiener processes} following recent efforts on policy optimization for LQ regulator problems~\cite{Fazel2018}; these systems may be modeled more accurately with uncertain additive Brownian noise where diffusion processes modeled with \^{I}to's stochastic calculus are the theoretical machinery for analysis. 
Prominent systems featuring such additive Wiener processes occur in economics and finance~\cite{SteeleStochCalc}, stock options trading~\cite{Oksendal}, protein kinetics, population growth models, and models involving computations with round-off error in floating point arithmetic calculations such as overparameterized neural network dynamics.

\textit{Our goal is to keep a controlled process, $z$, small in an infinite-horizon  constrained optimization setting under a minimizing policy $u \in \mc{U}$ in spite of unforeseen additive vector-valued stochastic Brownian process} $w(t) \in \mathbb{R}^q$ --- which may be of large noise intensity. In terms of the  $L_2$ norm, we can write $\|z\|_2 = \left(\int |z(t)|^2 dt\right)^{1/2}$. 
%
The associated performance criteria can be realized as minimizing the expected value of the risk-sensitive linear exponential functions of positive definite quadratic forms state and control variables 
\begin{align}
	\min_{u\in \mc{U}}\mc{J}_{exp}(x_0, u, w) := & 
	\mathbb{E}\bigg|_{x_0 \in \mc{P}_0}\exp\left[\frac{\alpha}{2}\int_{0}^{\infty}z^\top(t)z(t) \df t\right], \, \nonumber \\
	\text{subject to }	\df x(t) &= Ax(t) \df t + Bu(t) \df t + D \df w(t),  \nonumber \\
	z(t) &= Cx(t) + Eu(t), \,\, \alpha >0
	\label{eq:po_leqg_opt}
\end{align}
%
%
with state process $x \in \mathbb{R}^n$, output process $z \in \mathbb{R}^{p}$ to be controlled, and control input $u \in \mathbb{R}^{m}$. The derivative of $w(t) \in \bb{R}^{v} \, \ie dw/dt$ is a zero-mean Gaussian white noise with variance $W$, 
and $x(0)$ is a zero-mean Gaussian random vector independent of $w(t)$, $z(0)=0$, and $A\in \mathbb{R}^{n \times n}$, $B \in \mathbb{R}^{n \times m}, C  \in \mathbb{R}^{p \times n}, D \in \mathbb{R}^{n \times q}$, and $E \in\mathbb{R}^{p \times m}$ are constant matrix functions. The random signal $x(0)$ and the process $w(t)$ are defined over a complete probability space $(\Omega, \mc{F}, \mc{P})$. 
Suppose that we carry out a Taylor series expansion about $\alpha = 0$ in \eqref{eq:po_leqg_opt}, the variance term, $\alpha^{2} \texttt{var}(\int_0^\infty z^\top z)$, will be small after minimization. Thus, $\alpha$ can be seen as a measure of \textit{risk-aversion} if $\alpha >0$. It is important to note that in this paper, we only consider state feedback when $\alpha >0$. In particular when noise is present in the system, the value of $\alpha$ signifies the level of noise attenuation that penalizes the covariance matrix of the system's noise.

We adopt an adaptive policy optimzation policy iteration (PI) method in a continuous PO scheme. This can be seen as an instance of the actor-critic (AC) configuration in RL-based \textit{online} policy optimization schemes. Without explicit access to internal dynamics (system matrices), we iterate between steps of policy evaluation and policy improvement. Mimicking the actor in an RL AC setting, a parameterized controller must be evaluated relative to a parameterized cost function (the critic). The new policy is then used to improve the erstwhile (actor) policy by aiming to drive the cost to an extremal on the overall.

\textbf{Contributions}:
We focus on the more sophisticated case of optimizing an \textit{unknown stochastic linear policy} class $\mathcal{K}$ in an \textit{infinite-horizon} LQ cost setting such that optimization iterates enjoy the \textit{implicit regularization (IR) property}~\cite{Zhang2021}---satisfying $\hinf$ robustness constraints. We place PO for \textit{continuous-time linear stochastic controllers} on a rigorous \textit{global convergence} and \textit{robustness} footing. This is a distinguishing feature of our work. Our contributions are stated below.
\begin{itemize}
	\item We propose a \textit{two-loop iterative alternating best-response procedure} for computing the optimal \textit{mixed-design policy} parameterized by \textit{continuous time linear quadratic stochastic control}-- that accelerates the optimization scheme's convergence -- in model- and sampling-based cases;
	\item Rigorous convergence analyses follow for the model-based loop updates;
	\item In the absence of exact system models, we provide a robust PO scheme as a hybrid system with discrete-time samples from a nonlinear dynamical system. Its robustness is analyzed in an input-to-state (ISS) framework for robustness to perturbations and uncertainties  for loop updates.
	\item Lastly, we benchmark our results against the natural policy gradient~\cite{ShamNPG} in the spirit of recent system-theoretic analysis works~\cite{Fazel2018, Zhang2019, Bu2019, HuAnnualRevs}. 
\end{itemize}

\textbf{Notations}:
\label{sec:notations}
The set of all symmetric matrices with dimension $n$ is  $\mathbb{S}^n$ and $\bb{R}\, (\text{respectively } \mathbb{N}_+)$ is the set of real numbers (resp. positive integers). The Kronecker product is denoted by $\otimes$. The Euclidean (Frobenius) norm of a vector or the spectral norm of a matrix is $\norm{\cdot}$ ($\norm{\cdot}_F$). Let $\norm{\cdot}_\infty$ denote the supremum norm of a matrix-valued signals, i.e. $\norm{\Delta}_\infty= \sup_{s \in \mathbb{Z}_+}\norm{\Delta_s}_F$. The open ball of radius $\delta$ is $\mathcal{B}_\delta(X)= \{Y \in \bb{R}^{m \times n}| \norm{Y-X}_F < \delta \}$. The maximum and minimum singular values (eigenvalues) of a matrix $A$ are respectively denoted by $\sgmax(A)$ ($\eigmax(A)$) and $\sigma_{min}(A)$ ($\lambda_{min}(A)$). The eigenvalues of $A \in \bb{R}^{n\times n}$ are $\lambda_i(A)$ for $i=1, 2,\cdots, n$. For the transfer function $G(s)$, its $\mathcal{H}_\infty$ norm is  $\norm{G}_{\mathcal{H}_\infty} = \sup_{\omega \in \mathbb{R}} \sgmax(G(j\omega))$.

The $n$-dimensional identity matrix is $I_n$. 
The full vectorization of $X  \in \reline^{m\times n}$ is $\vec(X) = \left[x_{11}, x_{21}, \cdots, x_{m1}, x_{12}, \cdots, x_{m2}, \cdots, x_{mn}\right]^\top$. Let $P \in \bb{S}^n$; the half-vectorization of $P$ is the $n(n+1)/2$ column  vectorization of the upper-triangular part of $P$:  $\svec(P)=[p_{11}, \sqrt{2}p_{12}, \allowbreak \cdots, \sqrt{2}p_{1n}, p_{22}, \cdots, \sqrt{2}p_{n-1, n}, p_{nn}]^\top$. The vectorization of the dot product $\langle x, x^\top\rangle $, where $x
\in \ren$,  is  $\vecv(x)= [x_1^2, \cdots, x_1x_n, \allowbreak x_2^2, x_2x_3,\cdots,x_n^2]^\top$. 

\textbf{Paper Structure}:
A linear exponential quadratic Gaussian (LEQG) stochastic optimal control connection to dynamic games is first established in Section \ref{sec:back}. In \S \ref{sec:model_based},  we present a nested double-loop procedure for robust policy recovery in a mixed $\htwo$/$\hinf$ PO (in model-free and model-based) settings; this is followed by a rigorous analysis of their convergence and robustness properties. We  demonstrate the efficacy of our proposed algorithm on numerical examples, and discuss findings in \S \ref{sec:sim}.
	\section{PO: Dynamic Games Connection}
\label{sec:back}
%
In this section, we connect PO under linear controllers to the theory of two-person dynamic games. 
\begin{assumption}
	We take $C^\top C \triangleq Q\succ 0$, $E^T\left(C, \, E\right) = \left(0, \, R \right)$ for some matrix-valued function $R\succ 0$; and since in \eqref{eq:po_leqg_opt}, we want $w(t)$ to be statistically independent,  we take $D D^T$ = 0. Seeking a linear feedback controller for \eqref{eq:po_leqg_opt}, we require that the pair $(A, B)$ be \textit{stabilizable}. We expect to compute solutions via an optimization process, therefore we require that \textit{unstable modes of $A$ must be observable through $Q$}. Whence $(\sqrt{Q}, A)$ must be \textit{detectable}.
	\label{ass:realizability}
\end{assumption}
Given Assumption \autoref{ass:realizability}, the LEQG cost functional now becomes
\begin{align}
	\mc{J}_{exp}(x_0, u) = &
	\mathbb{E}\bigg|_{x_0 \in \mc{P}_0}\exp\left\{\frac{\alpha}{2}\int_{0}^{\infty}\left[x^\top (t)Q x(t) + \nonumber \right. \right. \\ & \left.\left.
	\qquad \qquad\qquad u^\top (t) R u(t)\right] \df t \right\},
	\label{eq:cost_leqg}
\end{align}
for a fixed $\alpha>0$ and the closed-loop transfer function is 
\begin{align}
	T_{zw}(K) =  \left(C-EK\right)(s\identity-A +B K)^{-1} D.
	\label{eq:leqg_tf}
\end{align} 
The set of all \textit{suboptimal} controllers that robustly stabilizes the linear system against all (finite gain) stable perturbations $\Sigma$, interconnected to the system by $w = \Sigma z$, such that $\|\Sigma\|_\infty \le 1/\gamma$ is
\begin{align}
	\mc{K} = \{\,K: \,\,\lambda_i(A-B_1K) < 0,  \,\, \|T_{zw}(K)\|_\infty < \gamma \}.
	\label{eq:constraint_leqg}
\end{align}


\begin{proposition}{\cite[Th. II.1]{Duncan2013}}
	The optimal control to the LEQG optimization  problem~\eqref{eq:po_leqg_opt} and cost functional ~\eqref{eq:cost_leqg} under $u$ in the infinite-horizon setting is of a linear-in-the-data form \ie $u^\star(t) = -{K}_{leqg}^\star \hat{x}(t) $
	%
	where gain $K_{leqg}^\star = R^{-1}B^\top P_\tau$,	and  $P_\tau$ is the unique, symmetric, positive definite solution to the algebraic Riccati equation (ARE)
	\begin{align} 
		A^\top P_\tau + P_\tau A - P_\tau(B R^{-1} B^\top - \alpha^{-2}DD^\top)P_\tau = - Q.
		\label{eq:GARE}
	\end{align}
\end{proposition}
\begin{corollary}
	In the infinite-horizon time-invariant case with constant system matrices and a stabilizable $(A, B)$, by the theorem on ``limit of monotonic operators"~\cite{Kantarovich} and~\cite[Theorem 9.7]{book_Basar}, we find that $ P^\star \triangleq P_\infty= \lim_{\tau \rightarrow \infty} P_{\tau}$, and $K_{leqg}^\star \triangleq K_\infty= \lim_{\tau \rightarrow \infty} K_{\tau}$.
\end{corollary}

\begin{remark}
	It is well-known that directly solving the LEQG problem~\eqref{eq:po_leqg_opt} in policy-gradient frameworks incurs biased gradient estimates during iterations; this may affect the preservation of risk-sensitivity in infinite-horizon LTI settings (see ~\cite{zhang2021derivative,Zhang2021}). As such, we introduce a workaround with an equivalent dynamic game formulation to the stochastic LQ PO control problem.
\end{remark}
	
	\begin{lemma}[Closed-loop Two-Player Game Connection]
		Consider the \textit{parameterized soft-constrained} upper value, with a stochastically perturbed noise process $w(t)$, which enters the system dynamics as an additive bounded Gaussian with known statistics\footnote{Since the time derivative of a Brownian process $w(t)$ is $dw(t)$, we maximize over the Gaussian $dw(t)$, rather than the unbounded stochastic noise $w(t)$.}, 
		\begin{align}
			\min_{u \in \mathcal{U}} \max_{\xi \in {W}}
			\mathcal{\bar{J}}_{\gamma}(x_0, u,&\xi) := \mathbb{E}\bigg|_{x_0 \sim \mc{P}_0, \, \xi(t)}\int_{0}^{\infty} \left[x^\top(t) Q x(t) + \right. \nonumber \\ &\left.
			\quad  u^\top(t) R u(t) - \gamma^2 \xi^\top(t) \xi(t)  \right]\df t\nonumber \\
			\text{subject to } \df x(t) &= Ax(t) \df t + Bu(t) \df t + D \xi(t), \nonumber \\
			z(t) &= C x(t) + E u(t)
			\label{eq:LQOptProb}
		\end{align}
		with $\xi (\equiv dw)$ as the zero-mean Gaussian noise with variance $W$(equivalent to $dw/dt$ in \eqref{eq:po_leqg_opt}), scalar $\gamma>0$ denoting the level of disturbance attenuation, and  $x_0$ an arbitrary initial state. Suppose that there exists a non-negative definite (nnd) solution of \eqref{eq:GARE} (with $\alpha$ replaced by $\gamma$), then its minimal realization, $P^\star$, exists.  
		If $(A, Q^{\frac{1}{2}})$ is observable, then every nnd solution  $P^\star$ of \eqref{eq:GARE} is positive definite. For a nnd $P^\star$, there exists a common upper and  lower value for the game and if $\bar{\mc{J}}_\gamma$ is finite for some $\gamma:=\hat{\gamma}>0$, then $\bar{\mc{J}}_\gamma$ is bounded (if and only if the pair $(A, B)$ is stabilizable) and equivalent to the lower value $\underline{\mc{J}}_\gamma$\footnote{The lower value is constructed by reversing the order of play in the value defined in \eqref{eq:LQOptProb}.}.  In addition, for a bounded $\bar{\mc{J}}_\gamma$ for some $\gamma=\hat{\gamma}$ and for  optimal gain matrices, $K^\star = R^{-1}B^\top P_{K,L}$, $L^\star = \gamma^{-2}D^\top P_{K,L}$,
		$\bar{\mc{J}}_\gamma$ admits the following Hurwitz feedback matrices for all $\gamma > \hat{\gamma}$
		\begin{align}
			A_K^\star = A-BK^\star , \, A^\star_{K,L} = A_K^\star + D L^\star
			\label{eq:CLMatrices}
		\end{align}
		where the nnd $P_{K,L}$ is the unique solution of \eqref{eq:GARE} for $\gamma > \hat{\gamma}$ in the class of nnd matrices if it renders $A^\star_{K,L}$ Hurwitz.  Whence, the saddle-point optimal controllers are
		\begin{align}
			u^\star(x(t)) = -K^\star x(t), \,\, \xi^\star (x(t)) = L^\star x(t).
			\label{eq:gains_zero_sum}
		\end{align}
		%
		\label{lm:DGSolution}
	\end{lemma}
	\begin{proof}
		The proof follows that in~\cite[Th. 9.7]{book_Basar} exactly if we preserve the $\gamma^{-1}$ term in the ARE of equation 9.31 in \cite{book_Basar} and replace it by $\gamma^{-2}$ as we have here.
	\end{proof}
	%
%
For any stabilizing control pair $(K, L)$, if \eqref{eq:gains_zero_sum} is applied to the system in \eqref{eq:LQOptProb}, the resulting cost from \eqref{eq:LQOptProb} is $\bar{J}_\gamma = (x_0^\top P_{K,L} x_0)$~\cite{Kleinman1968}.  In what follows, we present a double-loop iterative solver for the gains $K$ and $L$ in \eqref{eq:gains_zero_sum} -- in model-based and model-free settings. 
	\section{Policy Optimization via Policy Iteration}
\label{sec:model_based}
We now present a special case to Kleinman's policy iteration (PI) algorithm~\cite{Kleinman1968} in a PO setting via a nested double loop PI scheme  when (i) exact models are known; this will provide a barometer for our later analysis when (ii) exact system models are unknown. 

Let $p$ and $q$ be indices of nested iterations between updating the closed-loop minimizing player's controller $K_p$ (in an outer loop) and the maximizing player's controller $L_q(K_p)$ (in an inner-loop) for $p=\{1,2,\ldots, \bar{p}\}$ and $q=\{1,2,\ldots, \bar{q}\}$ for $(\bar{p}, \bar{q}) \in \bb{N}_+$. Furthermore, define the identities  
%
\begin{align}
&A_K^{p} = A-BK_p, \quad  A^{p,q}_{K,L} = A_K^p + D L_q(K_p) , \nonumber \\
& Q_{K}^{p} = Q + K_{p}^\top R K_{p},\,\, A_K^\gamma = A_K^p + \gamma^{-2} D D^\top P_K^p.
\label{eq:identities}
\end{align} 
%
For the soft-constrained value functional \eqref{eq:LQOptProb} at the $p$'th iterate of the minimizing controller $K$ we have the following value iteration form for \eqref{eq:GARE}, 
\begin{subequations}
\begin{align}
	A_K^{p\top} P_K^p &+ P_K^p A_K^p +Q_K^p  + \gamma^{-2} P^p_K DD^\top P^p_K  = 0, \label{eq:riccati_outer_loop_iter_ric} \\
	K_{p+1} &= R^{-1}B^\top P_K^p
	\label{eq:riccati_outer_loop_iter_gain}
\end{align}
\label{eq:riccati_outer_loop_iter}
\end{subequations}
where $P_K^p$ is the $p$'th iterate's solution to \eqref{eq:riccati_outer_loop_iter}.  Similarly, for the maximizing controller, $L_q(K_p)$, the following closed-loop continuous-time ARE (CARE) iteration applies
%
\begin{subequations}
\begin{align}
	&A_{K,L}^{(p,q)^\top} P^{p,q}_{K,L} + P^{p,q}_{K,L}  A_{K,L}^{p,q} + Q_K^p-  {\gamma^{2}} L_{q}^{\top}(K_p) L_{q}(K_p) =0  \\
	&K_{p+1} = R^{-1}B^\top P_K^{p,{q}}, \,\, L_{q+1}(K_p) = \gamma^{-2} D^\top P_{K,L}^{p,q}
	\label{eq:AREOuter} 
\end{align}
\label{eq:riccati_inner_loop_iter}
\end{subequations}
where $P_{K,L}^{p,q}$ is the solution to \eqref{eq:riccati_inner_loop_iter} for gains $[K_p,L_q(K_p)]$. Choosing a stabilizing minimizing player control gain, $K_p$ we first evaluate $u$'s  performance by solving \eqref{eq:riccati_outer_loop_iter}. \textit{This is the policy evaluation step in PI}. The policy is then improved in a following iteration by solving for the cost matrix in \eqref{eq:AREOuter} -- \textit{this is the policy improvement step}. The process can thus be seen as \textit{a policy iteration algorithm where the performance of an initial control gain $K_p$ is first evaluated against a cost function. A newer evaluation of the cost matrix $P_{K,L}^{p,q}$ is then used to improve the controller gain $K_{p+1}$ in the outer loop.}
\begin{problem}[Model-Based Policy Iteration]
Given system matrices $A, B, C, D, E$, find the optimal controller gains $K_p$, $L_q(K_p)$ that robustly stabilizes \eqref{eq:po_leqg_opt} such that the controller gains do not leave the set of all \textit{suboptimal} controllers denoted by
\begin{align}
	\breve{\mc{K}} &= \{(K_p,L_q(K_p)): \lambda_i(A_{K}^{p}) < 0,  \lambda_i(A_{K,L}^{p,q}) < 0, \,\nonumber \\
	&\qquad \|T_{zw}(K_p,L_q(K_p))\|_\infty < \gamma \text{ for all } (p, q) \in \bb{N}\}.
	\label{eq:robo_stable_cond}
\end{align}

\label{prob:mbased}
\end{problem}

\subsection{Double Loop (DL) Successive Substitution}
The procedure for obtaining the optimal $P^\star$ in Problem \ref{prob:mbased} is  described in Algorithm \ref{alg:model_based_design}. It  finds a global Nash Equilibrium (NE) (or equivalently a saddle-point equilibrium)~\cite{Zhang2019} of the LQ zero-sum game \eqref{eq:LQOptProb} by solving the nonlinear ARE \eqref{eq:riccati_inner_loop_iter} in a nested two-loop policy iteration (PI) scheme. 
\begin{algorithm}[t!] 
	\caption{(Model-Based) PO via Policy Iteration		
		\label{alg:model_based_design}}
	\KwIn{Max. outer iteration $\bar{p}$, $q=0$, and an $\epsilon >0$;}
	\KwIn{Desired risk attenuation level $\gamma>0$;}
	\KwData{Minimizing player's control matrix $R \succ 0$.}\
	Compute $(K_0, L_0) \in \mathcal{K}$; \Comment{{From \cite[Alg. 1]{LekanIFAC}}}\;
	Set  $P^{0,0}_{K,L} = Q^0_K$; \Comment{See equation  \eqref{eq:identities}}\;
	\For{$p=0,\ldots, \bar{p}$}
	{   
		Compute $Q^{p}_{K}$ and $A_K^{p}$ \Comment{See equation \eqref{eq:identities}}\;
		Obtain $P_K^p$ by evaluating $K_p$ on \eqref{eq:riccati_outer_loop_iter}\;
		\While{$\|P^{p}_{K} - P^{{p}, {q}}_{K,L}\|_F \le \epsilon$}
		{\label{alg:mb::inner_loop_start}
			Compute $L_{q+1}(K_p):= \gamma^{-2}D^\top P^{p,q}_{K,L}$\;
			\label{line:alg1::rl1}
			%
			Solve \eqref{eq:riccati_inner_loop_iter} until $\|P^{p}_{K} - P^{{p}, {q}}_{K,L}\|_F \le \epsilon$\;
			\label{alg:model_based_design::gare_solver}
			$\bar{q} \leftarrow q +1$
		}
		\label{alg:mb::inner_loop_end}
		Compute $K_{p+1} = R^{-1}B^\top P^{p,\bar{q}}_{K,L}$ \Comment{See \eqref{eq:AREOuter}} \;
		\label{alg:mb_based::gain_improv}
	}
\end{algorithm}

An initial $(K_0, L_0)$ control pair that guarantees the iterates'  feasibility upon projection onto the set $\breve{\mc{K}}$ is first determined in order to enforce the condition \eqref{eq:robo_stable_cond}. We refer readers to our recent conference paper~\cite{LekanIFAC} where this procedure is further elaborated\footnote{We remark that this can also be found by solving a convex optimization problem if the Riccati equation is formulated as a linear matrix inequality~\cite{boyd1994linear}.}. Afterwards,  the Riccati equation's \eqref{eq:riccati_inner_loop_iter} solution \ie $P_{K,L}^{0,0} \triangleq Q_K^0$ must be computed and the gains $[K_p, L_q(K_p)]$ are updated as in \eqref{eq:AREOuter}. \textit{As an RL-based PO procedure, Line \ref{alg:mb_based::gain_improv} in Alg. \ref{alg:model_based_design}  can be seen as a reinforcement over the iterates $\{p, p+1\}$}. 
After the last update of the inner loop maximizing player gain $L_{\bar{q}}(K_p)$, the outer-loop update of the minimizing controller gain $K_p$ robustly stabilizes the closed-loop transfer function from $w$ to $z$ \ie $T_{zw}(K_p,L_q(K_p))$ under gains $K_p$ and $L_{\bar{q}}(K_p)$ against all (finite gain) disturbances $\Sigma$ interconnected to system \eqref{eq:LQOptProb} (by $w = \Sigma z$) such that $\|\Sigma\|_\infty \le \gamma^{-2}$. 

\subsection{Outer Loop Stability, Optimality, and Convergence}

We now  discuss the convergence guarantees of the iterations under perfect dynamics. 
Let us first introduce the preliminary results. 
%

%
\begin{lemma}
Under Assumption \ref{ass:realizability} and for the ARE \eqref{eq:riccati_outer_loop_iter}, if $K_0 \in \mathcal{K}$\footnote{The $\mc{K}$ defined here refers to the one defined in \eqref{eq:constraint_leqg}.}, then for any $p \in \mathbb{N}_+$, we must have the following conditions for the optimal $K^\star$ and $P^\star$,
\begin{enumerate}[(1)]
	\item $K_p \in \mathcal{K}$; \label{lm:cvg::Kp}
	\item $P_{K}^0 \succeq P_{K}^1 \succeq \cdots P_{K}^p \succeq \cdots \succeq P^{\star}$; 
	\item $\lim_{p \rightarrow \infty}\norm{K_p - K^*}_F = 0$, $\lim_{p \rightarrow \infty}\norm{P_{K}^p - P^*}_F = 0$.
\end{enumerate}
\label{lm:oloop_iter}
\end{lemma}

We provide a straightforward proof (with the robustness operator $\gamma$ of \eqref{eq:riccati_outer_loop_iter}) in the appendix. An alternate proof is given in \cite{Kleinman1968}. 
%
%
In \cite[Theorem A.7 and A.8]{Zhang2021}, the authors showed that the controller update phase in the outer-loop has a global sub-linear and local quadratic convergence rates. We now demonstrate that the outer-loop iteration has a global linear convergence rate. Let us first establish a few preliminary results that we will need in the proof of our main result.
\begin{lemma}
Let $\Psi =(K_{p+1}- K_p)^\top R (K_{p+1}-K_p)$; and $\Psi= \Psi^\top \succeq 0$. Furthermore, let $\Phi \in \mathbb{R}^{n \times n}$ be Hurwitz so that  $\Theta = \int_{0}^{\infty} e^{(\Phi^\top t)} \Psi e^{(\Phi t)} dt$ and define $c(\Phi) = {\log(5/4)}{\norm{\Phi}}^{-1}$. Then, $\norm{\Theta} \geq \frac{1}{2} c(\Phi) \norm{\Psi}.$
\label{lm:mb_upper_bd}
\end{lemma}

\begin{proof}
Define $S(t) ={\sum_{k=1}^{\infty} (\Phi t)^k}/{k!}$ so that $	e^{\Phi t} = I_n +{\sum_{k=1}^{\infty} (\Phi t)^k}/{k!} \triangleq I_n + S(t)$ after  a Taylor series expansion.
Whence $	\norm{S(t)} = \sum_{k=1}^{\infty} (\norm{\Phi} t)^k/{k!}$ or $\norm{S(t)} \ge e^{\norm{\Phi}t} - 1$.
	%
	For $x_0 \neq 0$ satisfying $x_0^\top \Psi  x_0 = \norm{\Psi}\norm{x_0}^2$, it can be verified that
	\begin{align}
			x_0^\top \Theta x_0 & = \int_{0}^{\infty}x_0^\top e^{\Phi^\top} \Psi e^{\Phi}x_0 \df t \geq \int_{0}^{c(\Phi)}x_0^\top e^{\Phi^\top} \Psi e^{\Phi}x_0 \df t, \nonumber \\
			&\ge \int_{0}^{c(\Phi)} \frac{1}{2}\norm{\Psi}\norm{x_0}^2 \df t 
			\geq \frac{1}{2}c(\Phi)\norm{\Psi}\norm{x_0}^2.
		\label{eq:bound_on_E}
	\end{align}
	\textit{A fortiori}, Lemma \ref{lm:mb_upper_bd}'s proof follows from \eqref{eq:bound_on_E}.
\end{proof}

\begin{remark}
	For $A_K = A-BK$, we know from the bounded real Lemma~\cite[Lemma A.1]{Zhang2021} that the Riccati equation 
	\begin{align}
		A_K^\top P_K + P_K A_K + Q_K + \gamma^{-2} P_K DD^\top P_K = 0
		\label{eq:brlemma2}
	\end{align}
	admits a unique positive definite solution $P_K \succ 0$ with a Hurwitz  $(A_K + \gamma^{-2} DD^\top P_K)$.
\end{remark}

\begin{lemma}[Optimality of the iteration]
	Consider any $K \in \mathcal{K}$, let $K^\prime= R^{-1}B^\top P_K$ (where $P_K$ is the solution to \eqref{eq:brlemma2}, and $\Psi_K= (K - K^\prime)^\top R(K-K^\prime)$. If $\Psi_K=0$, then $K = K^\star$.
	\label{lm:EK=0}
\end{lemma}

\begin{lemma}[Bound on Cost Difference Matrix]
	For any $h>0$, define $\mathcal{K}_h := \{K \in \mathcal{K}| \Tr(P_K^p - P^\star) \leq h \}$. For any $K \in \mathcal{K}_h$, let $K^\prime := R^{-1}B^\top P_K^p$, where $P_K^p$ is the $p$'th iterate solution to \eqref{eq:brlemma2}, and $\Psi_{K_p}= (K_p - K_p^\prime)^\top R(K_p-K_p^\prime)$. Then, there exists $b(h)>0$, such that $\norm{P_K^p - P^\star}_F \leq b(h) \norm{\Psi_{K_p}}_F$.
	\label{lm:bound_EK}
\end{lemma}
\begin{theorem}
	For any $h>0$ and $K_0 \in \mathcal{K}_h$, there exists $\alpha(h) \in \bb{R}$ 
	such that $\Tr(P_{K}^{p+1} - P^\star) \leq \alpha(h) \Tr(P_{K}^p - P^\star)$. That is, $P^\star$ is an exponentially stable equilibrium.
	\label{thm:oloop_converge_rate}
\end{theorem}
\begin{proof}
	Since $[P_{K}^p - P_{K}^{p+1}]$ satisfies the Lyapunov equation \eqref{eq:riccati_outer_converge_lm1.4}, $(K_{p+1} - K_p)^\top R (K_{p+1}-K_p)\ge 0$ implies that $P_{K}^p - P_{K}^{p+1}\succeq 0$ by Lemma \ref{lm:ZhouRobust}. Hence  $A_{K}^{(p+1)}$ must be Hurwitz going by Lemma~\ref{lm:inverseLya}. Define $H_{K}^p=\int_{0}^{\infty} e^{A_K^{(p+1)^\top}t}  \Psi_p e^{A_K^{(p+1)}t}\df t$ where $\Psi_p$ is as defined in Lemma \ref{lm:mb_upper_bd}. Thus, we have by the second statement of Lemma~\ref{lm:inverseLya} that the cost matrix admits the form (see statement 1 of Lemma \ref{lm:ZhouRobust})
	\begin{align}\label{eq:PKidiff}
		\begin{split}
			P_{K}^p - P_{K}^{p+1} \succeq H_{K}^p.
		\end{split}
	\end{align}
	From Lemma \ref{lm:oloop_iter}, we have for a $p>0$ that $P_{K}^0 \succeq P_{K}^p$ so that for an $h \ge 0$, we find that %
	\begin{align}
		\norm{A_{K}^{p+1}} \leq \norm{A} + (\norm{BR^{-1}B^\top}+\gamma^{-2}\norm{DD^\top})h.
	\end{align} 
	Set $c(h) = {\log(5/4)}/{\norm{A} + (\norm{BR^{-1}B^\top}+\gamma^{-2}\norm{DD^\top})h}$
	so that we have $\norm{ H_{K}^p} \geq \frac{1}{2} c(h) \norm{\Psi_p}$  from Lemma \ref{lm:mb_upper_bd}. Using Lemmas \ref{lm:bound_EK} and \ref{lm:normTrace}, and taking the trace of \eqref{eq:PKidiff} we find that
	\begin{align}
		&\Tr(P_{K}^{p+1} - P^\star) 
		\leq \Tr(P_{K}^p - P^\star) - \Tr(H^p_{K}), \nonumber \\
		&\leq \Tr(P_{K}^p - P^\star) - {c(h)}\norm{\Psi_p}/2, \nonumber \\
		&\leq \Tr(P_{K}^p - P^\star) - \frac{\sqrt{n}c(h)}{2n}\norm{\Psi_p}_F, \nonumber \\
		&\leq \Tr(P_{K}^p - P^\star) - \frac{c(h)}{2\sqrt{n}b(h)}\norm{P_{K}^p - P^\star}_F, \nonumber \\
		&\leq \left(1-\frac{c(h)}{2 n b(h)}\right)\Tr(P_{K}^p - P^\star).
		\label{eq:trace_bound}       
	\end{align}
	The proof follows if we set $\alpha(h)=1-c(h)/2 n b(h)$.
\end{proof}


\subsection{Inner Loop Stability, Optimality, and Convergence}
We now analyze the monotonic convergence rate of the inner loop. Given arbitrary gains $K_p \in \mathcal{K}$  and $L_q (K_p)$, let $P_{K,L}^{p,q}$ be the positive definite solution of the associated Lyapunov equation \eqref{eq:riccati_inner_loop_iter}.  
The following lemma shows that the cost matrix $P_{K,L}^{p,q}$ monotonically converges to \eqref{eq:riccati_inner_loop_iter}'s solution. 

\begin{lemma}
	Suppose that $L_0(K_0)$ is stabilizing, then for any $q \in \mathbb{N}_+$ (with $P_{K,L}^{p,\bar{q}}$ as the solution to \eqref{eq:riccati_inner_loop_iter}), 
	\begin{enumerate}
		\item $A_{K,L}^{p,q}$ is Hurwitz;
		\item $P_{K,L}^{p,\bar{q}} \succeq \cdots \succeq P_{K}^{(p, q+1)} \succeq P_{K}^{p, q} \succeq \cdots \succeq P_{K,L}^{p,0}$; and 
		\item $\lim_{q \rightarrow \infty} \norm{P_{K,L}^{p,q} - P_{K,L}^{p,\bar{q}}}_F = 0$.
	\end{enumerate}
	\label{lm:inner_loop_iter}
\end{lemma}
A proof is provided in the Appendix. 
%
%
We next analyze the monotonic convergence of the inner loop of the nested double loop algorithm. Let us first discuss a  preliminary result.

\begin{lemma}[Monotonic Convergence of the Inner-Loop]
	For any $K \in \mathcal{K}$, let $L(K)$ be the control gain for the player $w$ such that $A_K + DL(K)$ is Hurwitz. Let $P_K^L$ be the solution of
	\begin{align}
		&\left(A_K+DL(K)\right)^\top P_K^L + P_K^L\left(A_K+DL(K)\right) + Q_K \nonumber \\
		&\qquad \qquad -\gamma^2L(K)^\top L(K) = 0.
		\label{eq:lyapunovPKL}
	\end{align}
	Let $L^\prime(K) = \gamma^{-2}D^\top P_K^L$ and $\Psi_K^L = \gamma^{-2} (L^\prime(K) - L(K))^\top (L^\prime(K) -L(K))$.
	Then, for a $c(K)= \Tr\left(\int_{0}^{\infty} e^{\left(A_K + DL(K^\star)\right)t} e^{\left(A_K + DL(K^\star)\right)^\top t} \mathrm{d}t\right)$, the following inequality holds $\Tr(P_K - P_K^L) \leq \norm{\Psi_K^L}c(K)$.
	%
	\label{lm:traceInner}
\end{lemma}

\begin{theorem}
	For a $K \in \breve{\mathcal{K}}$, and for any $(p,q) \in \mathbb{N}_+$, there exists $\beta(K) \in \bb{R}$ 
	 such that 
	\begin{align} 
		\Tr(P_{K}^p - P_{K,L}^{p,q+1}) \leq \beta(K) \Tr(P_K^p - P_{K,L}^{p,q}).
		\label{eq: linear conv}
	\end{align}
	\label{thm:iloop_converge_rate}
\end{theorem}

\begin{proof}
	Define $\Psi_K^q = \gamma^2 \left[L_{q+1}  - L_q\right]^\top \left[L_{q+1}  - L_q\right]$ in \eqref{eq:riccati_inner_diff} and $F_K^q = \int_{0}^{\infty} e^{(A_{K,L}^{(p,q+1)^\top})t}\Psi_K^q e^{A_{K,L}^{(p,q+1)t}} \mathrm{d}t$. By Lemma \ref{lm:inner_loop_iter},  $A_{K,L}^{p,q+1}$ is Hurwitz in \eqref{eq:riccati_inner_diff} so that from Lemma \ref{lm:ZhouRobust}  we have 
	\begin{align}
		P_{K,L}^{p,q+1} - P_{K,L}^{p,q} = F_K^q.  
		\label{eq:thm_inner}
	\end{align}
	By Lemma \ref{lm:inner_loop_iter}, $P_{K,L}^{p,q+1} \succeq P_{K,L}^{p,q}$ so that
	\begin{align}\label{eq:P_Knner_loop}
		\norm{A_{K,L}^{(p,q+1)}} \leq \norm{A-BK_p} + \gamma^{-2}\norm{DD^\top}\norm{P_{K}^p}.
	\end{align}
	\begin{align}\label{eq:defdK}
		\text{Let }	d(K) = {\log(5/4)}/\left({\norm{A_K} + \gamma^{-2}\norm{DD^\top}\norm{P_{K}^p}}\right), 
	\end{align}
	so that subtracting both sides of \eqref{eq:thm_inner} from $P_K^p$ and taking the trace of the resulting expression, we find that
	\begin{subequations}
		\begin{align}
			\Tr(P_K^p - &P_{K,L}^{p,q+1}) =  \Tr(P_K^p - P_{K,L}^{p,q}) - \Tr(F_K^q), \\
			&\leq \Tr(P_K^p - P_{K,L}^{p,q}) - \norm{F_K^q},
			\label{eq:normF}\\
			&\leq \Tr(P_K^p - P_{K,L}^{p,q}) - \frac{1}{2}d(K)\norm{\Psi_K^q} \label{eq:normFnormE},
		\end{align}
	\end{subequations}
	where we have used Lemma \ref{lm:normTrace} to arrive at the inequality in \eqref{eq:normF}, and extended Lemma \ref{lm:mb_upper_bd} to arrive at \eqref{eq:normFnormE} since $\Psi_K^q = \Psi_K^{q^\top}$. 
	Furthermore, from Lemma \ref{lm:traceInner} 
	\begin{align}
		\Tr(P_K^p -P_{K,L}^{p,q}) &\leq   \left(1 - \dfrac{d(K)}{2c(K)}\right)\Tr(P_K^p - P_{K,L}^{p,q}).
		\label{eq: linear conv proof}
	\end{align} 
	The proof follows if we set $\beta(K) = 1 -{d(K)}/{2c(K)}$.
\end{proof}

\begin{remark}
	As seen from Lemma \ref{lm:inner_loop_iter}, $P_K^p - P_{K,L}^{p,q} \succeq 0$. From Lemma \ref{lm:normTrace} and the result of Theorem \ref{thm:iloop_converge_rate}, we have $\norm{P_K - P_{K,L}^{p,q}}_F \leq \Tr(P_K - P_{K,L}^{p,q}) \leq \beta(K)\Tr(P_K)$, i.e. $P_{K,L}^{p,q}$ exponentially converges to $P_K$ in the  Frobenius norm.
\end{remark}

%
\begin{lemma}[Uniform Convergence of Iterates]
	For any $h>0$, $K \in \mathcal{K}_h$, and $\epsilon>0$, there exists ${q}^\prime(h) \in \mathbb{N}_+$ independent of $K$, such that if $q \geq {q}^\prime(h)$, $\norm{P_{K,L}^{p,q} - P_{K}}_F \leq \epsilon$.
	\label{thm:uniform_converge}
\end{lemma}
\begin{proof}[Proof of Lemma \ref{thm:uniform_converge}]
	This Lemma is an immediate outcome of Theorems \ref{thm:oloop_converge_rate} and \ref{thm:iloop_converge_rate}.
\end{proof}

\subsection{Sampling-based PO on Hybrid Discrete-Time Nonlinear System} 
\label{ssec:hybrid}
	The exact knowledge of the system matrices $A, B, C, D, E$ are often  unavailable so that the policy evaluation step will result in biased estimates. When errors are present from using I/O or state data for the PO procedure in Alg. \ref{alg:model_based_design}, residuals from early termination of numerically solving Line \ref{alg:model_based_design::gare_solver} in Alg. \ref{alg:model_based_design}, or using an approximate cost function owing to inexact values of $Q$ and $R$, the algorithm may fail to converge.  

\begin{problem}[Sampling-based Policy Optimization]
	If $A, B, C, D, E, P$ are all replaced by approximate matrices $\hat{A}, \hat{B}, \hat{C}, \hat{D}, \hat{E}, \hat{P}$, under what conditions will the sequences $\{\hat{P}_{K,L}^{p,q}\}_{(p,q)=1}^{(p,q)=\infty}$, $\{\hat{K}_p\}_{p=0}^\infty$, $\{\hat{L}_q\}_{q=0}^\infty$ converge to a small neighborhood of the optimal values $\{P_{K,L}^\star\}_{(p,q)=0}^{(p,q)=\infty}$, $\{{K}^\star_p\}_{p=0}^\infty$, and $\{{L}^\star_q\}_{q=0}^\infty$.
	\label{prob:mfree}
\end{problem}

\subsubsection{Discrete-Time Nonlinear System Interpretation}
From Assumption \ref{ass:realizability}, a $P_K^0 \in \bb{S}^n$ exists such that when applied to find $K_0$ \ie $K_0 = R^{-1} B^\top P^0_K$, such a $K_0$ will be stabilizing. 
Now, factoring in approximation errors between the policy evaluation and improvement steps, we end up with a hybrid system consisting of a  continuous-time policy gain pair $(\hat{K}_p, \hat{L}_q(\hat{K}_p))$ and a learning algorithm that is essentially a discrete sampled data from a nonlinear system (owing to errors from various sources). We will leverage Lemmas \ref{lm:oloop_iter} and \ref{lm:inner_loop_iter} to show that under inexact loop updates and lumping gain iterate estimate errors as system inputs to the online PO scheme, it converges to the optimal solution and closed-loop dynamic stability is guaranteed in an input-to-state stability framework (ISS) \cite{Sontag2008}. Hence the loops are discrete-time nonlinear systems.

\subsubsection{Online (Model-Free) Nested Loop Reparameterization}

Consider \eqref{eq:riccati_outer_loop_iter_gain} and suppose that $\hat{P}_K^0 \in \bb{S}^n$ is chosen following Assumption \ref{ass:realizability}. It follows that a $\hat{K}_k^1 = R^{-1}B^\top \hat{P}_K^0$ will be stabilizing since $\tilde{K}_k^1 \equiv \hat{K}_k^1 - {K}_k^1  \triangleq 0$. The same argument applies for $L_0$. For $(p,q)>0$, we must show that for  $\tilde{K}_k^p \equiv \hat{K}_k^p - {K}_k^p  \triangleq 0$ so that the sequence $\{P_{K,L}^{p,q}\}_{(p,q)=0}^\infty$ will converge to the locally exponentially stable $\hat{P}_{K,L}^\star$ going by Lemmas \ref{lm:oloop_iter} and \ref{lm:inner_loop_iter}. We proceed by lumping the estimate errors as an input into the gain terms to be computed in the PO algorithm.
Under inexact outer loop update, the iterate $K_{p+1}$ becomes inaccurate so that the inexact outer-loop iteration involves the recursions
\begin{subequations}
	\begin{align}
		& \hat{A}_K^{p\top} \hat{P}_K^p +  \hat{P}_K^p \hat{A}_K^{p} + \hat{Q}_K^p+ \gamma^{-2} \hat{P}_K^p D D^\top \hat{P}_K^p = 0, \\
		& \hat{K}_{p+1} =  R^{-1} B^\top \hat{P}_K^p  + \tilde{K}_{p+1} \triangleq \bar{K}_{p+1} + \tilde{K}_{p+1},
		\label{eq:inexact_iter_oloop_gain}
	\end{align}
	\label{eq:inexact_iter_oloop}
\end{subequations}
where $\hat{A}_K^{p} = A - B\hat{K}_p$ and $\hat{Q}_K^p = Q + \hat{K}_p^\top R \hat{K}_p$. Similar argument applies to the inner loop updates so that the inexact inner loop update is
\begin{subequations}
\begin{align}
	&\hat{A}_{K,L}^{p,q\top} \hat{P}^{p,q}_{K,L} + \hat{P}^{p,q}_{K,L}  \hat{A}_{K,L}^{p,q} + \hat{Q}_K^p- {\gamma^{2}} \hat{L}_{q}^{\top} \hat{L}_{q}(\hat{K}_p) =0 \label{eq:inexact_iter_iloop_dyna} \\
	&\hat{K}_{p+1} = R^{-1}B^\top \hat{P}_K^{p,{q}} + \tilde{K}_p, \\
	&\hat{L}_{q+1}(\hat{K}_p) = \gamma^{-2} D^\top \hat{P}_{K,L}^{p,q} + \tilde{L}_{q+1}(\tilde{K}_p) \\
	& \qquad \qquad \triangleq \bar{L}_{q+1}(\bar{K}_p) + \tilde{L}_{q+1}(\tilde{K}_p)
	\label{eq:inexact_iter_iloop_gain} 
\end{align}
\label{eq:inexact_iter_iloop}
\end{subequations}

Consider the transformation of the infinite-dimensional stochastic differential equation \eqref{eq:po_leqg_opt} in light of the identities \eqref{eq:identities} under inexact updates for $(p,q)>0$
\begin{align}
	\df x = [\hat{A}_{K,L}^{p,q} x  + B (\hat{K}_p x - D\hat{L}_q(K_p) + u)]\df t + D \df w.
	\label{eq:inf-dim-trans}
\end{align}
On a time interval $[s, s+\delta s]$, it follows from It{\^o}'s stochastic calculus and the Hamilton-Jacobi-Bellman equation that
\begin{align}
	&d\left[x^\top(s+\delta s) \hat{P}_{K,L}^{p,q} (s+\delta s) - x^\top(s) \hat{P}_{K,L}^{p,q} x(s)\right]= \nonumber \\
	& \,\, (\df x)^\top \hat{P}_{K,L}^{p,q}  x + x^\top \hat{P}_{K,L}^{p,q} \df x + (\df x)^\top \hat{P}_{K,L}^{p,q}  (\df x). 
\end{align}
Along the trajectories of equation \eqref{eq:inf-dim-trans} and using the gains in  \eqref{eq:riccati_inner_loop_iter}, the \rhs in the foregoing becomes
%
\begin{align}
	&x^\top \left[\hat{A}_{K,L}^{p,q \top} \hat{P}_{K,L}^{p,q} +\hat{P}_{K,L}^{p,q} \hat{A}_{K,L}^{p,q}\right] x \df t + 2 x^\top \hat{P}_{K,L}^{p,q} D \df w 
	\label{eq:sys_trajos_mats}  \\
	& + 2 x^\top \hat{P}_{K,L}^{p,q} B (K_p x - D\hat{L}_q(K_p) + u) dt + \Tr(D^\top P D), \nonumber \\
	%
	%
&\text{so that }	x^\top(s+\delta s) \hat{P}_{K,L}^{p,q} (s+\delta s) - x^\top(s) \hat{P}_{K,L}^{p,q} x(s) \nonumber \\
	&=\int_s^{s+\delta s}\left[(-x^\top \hat{Q}_K^p x - \gamma^{2}w^\top w)\df t   + 2 \gamma^2 x^\top \hat{L}^\top_{q+1}(K_p) \df w\right]  \nonumber \\
	& \,+\int_s^{s+\delta s} 2 x^\top \hat{K}_{p+1}^\top R \left[\hat{K}_p x - D \hat{L}_q(\hat{K}_p) + u\right]\df t \nonumber \\
	& \qquad + \int_s^{s+\delta s}\Tr(D^\top \hat{P}_{K,L}^{p,q} D) \df t.
	\label{eq:sys_trajos}
\end{align}
\textbf{Observe}: The system dependent  matrices $\hat{A}_{K,L}^{p,q}, B, C, D$ from equation \eqref{eq:sys_trajos_mats} are now replaced by input and state terms including $\hat{Q}_K^p$, $\hat{K}_{p+1}$, and $\hat{L}_{q+1}$ which are all retrievable via online measurements. \textit{We essentially end up with an input-to-state system}. The price we pay is that the noise feedthrough matrix $D$ must be known precisely. In this article, as is common in many linear stochastic system with Brownian motion, $D$ is taken to be identity~\cite{DuncanSDEBrownian, DuncanStochastic}.

\subsubsection{Sampling-based PO Scheme}
Our goal is to explore the system model until exact equality of $\hat{A}_{K,L}^{p,q}, \hat{P}_{K,L}^{p,q}$ and $\hat{K}_{p+1}, \hat{L}_{q+1}(K_p)$ to the corresponding terms in \eqref{eq:riccati_inner_loop_iter} occur. To this end, \eqref{eq:sys_trajos} allows us to explore with the controls $ u = -K_0 x + \eta_p$ and $w = -L_0 x + \eta_q$ where $(\eta_p, \eta_q)$ is drawn uniformly at random over matrices with a Frobenium norm $r$ similar to~\cite{Gravell2021, Fazel2018}. Let us now introduce the following identities,
\begin{align}
	&x^\top \hat{Q}_K^p x = (x^\top \otimes x^\top) \, \vec(\hat{Q}_K^p ), \nonumber \\
	&\gamma^{2}w^\top w = \gamma^{2}(w^\top \otimes w^\top) \,\,\vec(I_v),\nonumber \\
	& 2 \gamma^2 x^\top \hat{L}_{q+1}^\top (\hat{K}_p)\df w =  2 \gamma^2 (\identity_n \otimes x^\top) \df w \,\, \vec (\hat{L}_{q+1}^\top (\hat{K}_p)), \nonumber \\
	&2 x^\top \hat{K}_{p+1}^\top R\hat{K}_p x = 2 (x^\top \otimes x^\top) (I_n \otimes  \hat{K}^\top_p) \,\vec (\hat{K}_{p+1}^\top R), \nonumber \\
	&2 x^\top \hat{K}_{p+1}^\top R D \hat{L}_{q}(\hat{K}_p) = 2 (\hat{L}^\top_q(\hat{K}_p)D^\top \otimes x^\top)\, \vec (\hat{K}_{p+1}^\top R), \nonumber \\
		&2 x^\top \hat{K}_{p+1}^\top R u = 2 (u^\top \otimes x^\top)\, \vec (\hat{K}_{p+1}^\top R), \nonumber \\
		&\Tr(D^\top \hat{P}_{K,L}^{p,q} D)  = \vec^\top(D)\, \vec(\hat{P}_{K,L}^{p,q} D).
\end{align}
Furthermore, consider the matrices $\delta_{xx} \in \bb{R}^{\frac{n(n+1)}{2}l}$,  $\delta_{ww} \in \bb{R}^{\frac{v(v+1)}{2}l}$, $\identity_{xx} \in \bb{R}^{l\times n^2}$, and $\identity_{ux} \in \mathbb{R}^{l\times nm}$ for  $l \in \bb{N}_+$ so that
\begin{align}
	\Delta_{xx} &= \left[\vecv (x_1),  \ldots, \vecv (x_l)\right]^\top, \,\,  x_l = x_{l+1} - x_{l}, 
	\nonumber \\
	\Delta_{ww} &= \left[\vecv (w_1),  \ldots, \vecv (w_l)\right]^\top, \,\,  w_l = w_{l+1} - w_{l},\nonumber \\
	I_{xx} &= \left[\int_{s_0}^{s_{1}}x \otimes x \, \df t, \ldots, \int_{s_{l-1}}^{s_{l}}x \otimes x \, \df t\right]^\top, \nonumber \\
	I_{ww} &= \left[\int_{s_0}^{s_{1}}w \otimes w \, \df t, \ldots, \int_{s_{l-1}}^{s_{l}}w \otimes w \, \df t\right]^\top, \nonumber \\
	I_{xw} & = \left[\int_{s_0}^{s_{1}} (\identity_n \otimes x) \df w,  \ldots, \int_{s_{l-1}}^{s_{l}} \, (\identity_n \otimes x) \df w\right]^\top, \nonumber \\
	I_{ux} &= \left[\int_{s_0}^{s_{1}}u \otimes x \, \df t, \ldots, \int_{s_{l-1}}^{s_{l}}u \otimes x \, \df t\right]^\top.
\end{align}
Next, set 
\begin{subequations}
	\begin{align}
		\Theta_{K,L}^{p,q} &= \left[\Delta_{xx}, -2 I_{xx}(\identity_n \otimes \hat{K}_p^\top) + 2 (\hat{L}^\top_q(\hat{K}_p) D^\top \otimes x^\top)  \nonumber \right. \\ & \left.
		\,  - 2 I_{ux},  -2\gamma^2 I_{xw}, -\vec^\top(D)\vec(\hat{P}_{K,L}^{p,q}D) \right], \\
		\Upsilon_{K,L}^{p,q} &= \left[-\identity_{xx} \vec(\hat{Q}_K^p),\,\, -\gamma^2 \identity_{ww} \vec(\identity_v)\right].
	\end{align}
\end{subequations}
Define $\mathbf{1}_{q^2}$ as one-vector with dimension $q^2$. Thus, 
\begin{align}
	\Theta_{K,L}^{p,q} &\begin{bmatrix}
		\svec({P}_{K,L}^{p,q}) & \vec(\hat{K}_{p+1}^\top R) & \vec(\hat{L}_{q+1}^\top(\hat{K}_p)) & \mathbf{1}_{q^2}
	\end{bmatrix}^\top \nonumber \\
	& \qquad \qquad = \Upsilon_{K,L}^{p,q}.
\end{align}
Suppose that $\Theta_{K,L}^{p,q}$ is of full rank, then we can retrieve the unknown matrices via least squares estimation \ie 
\begin{align}
	\begin{bmatrix}
		\svec({P}_{K,L}^{p,q}) \\ \vec(\hat{K}_{p+1}^\top R) \\ \vec(\hat{L}_{q+1}^\top(\hat{K}_p))\df w \\ \mathbf{1}_{q^2}
	\end{bmatrix} = (\Theta_{K,L}^{p,q\top} \Theta_{K,L}^{p,q})^{-1} \Theta_{K,L}^{p,q\top} \Upsilon_{K,L}^{p,q}.
	\label{eq:model_free}
\end{align}
We thus end up with a scheme for retrieving the system matrices provided that the algorithm is robust to perturbations upon iterating through \eqref{eq:model_free} for each $(p,q)$. The full scheme is summarized in the flowchart of \autoref{fig:sampling_based}. We next state the condition under which $\Theta_{K,L}^{p,q}$ is of full rank.
\begin{lemma}{\cite[Lemma 6]{JiangJiang}}
	If there exists an integer $l_0>0$ such that for all $l \ge l_0$,  $rank(I_{xx}, I_{ux}, I_{xw}, \mathbf{1}_{q^2}) = n(n+1) + mn +  nq + q^2$, then $\Theta_{K,L}^{p,q}$ has full rank for all $(p,q) \in (\bar{p}, \bar{q})$.
	\label{lm:rank_cond}
\end{lemma}
\begin{remark}
	Lemma \ref{lm:rank_cond} allows the convergence assurance of \autoref{fig:sampling_based} under the condition that the rank condition be fulfilled.
\end{remark}
%
\begin{figure}[tb]
	\centering 
	\includegraphics[width=\columnwidth]{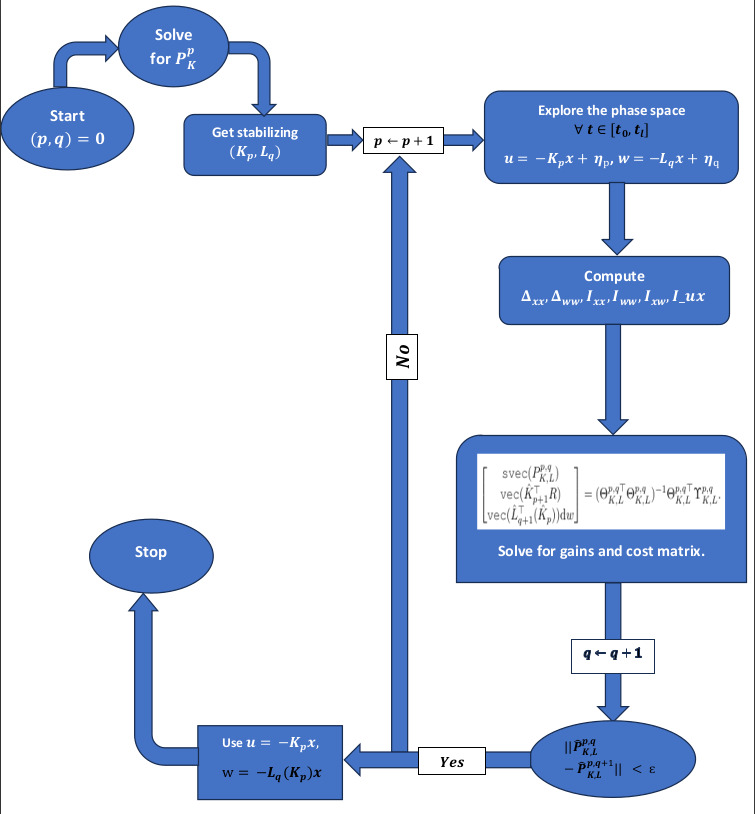}
	\caption{Flowchart for Sampling-based PO in Continuous-time Mixed $\htwo/\hinf$ Stochastic Control.}
	\label{fig:sampling_based}
\end{figure}
\subsubsection{Robustness of Minimizing Controller to Perturbations}
We now analyze the robustness of the sampling-based scheme as a hybrid nonlinear discrete time system gains with continuous-time dynamics. Let $\tilde{P} = P_K - \hat{P}_K$ and $\tilde{K} = K - \hat{K}$ denote errors arising from the inexact updates.
\begin{lemma}[Outer-Loop Robustness to Perturbations]
	For any $K \in \mathcal{K}$, there exists an $e(K) >0$ such that for a perturbation  $\tilde{K}$, $K + \tilde{K} \in \mathcal{K}$, as long as $\norm{\tilde{K}} < e(K)$.
	\label{lm:robust_after_perturb}
\end{lemma}
As long as $\tilde{K}$ is small, if we start with a robustly stabilizing $K \in \mathcal{K}$, we can guarantee the feasibility of the iterates. 
\begin{theorem} \label{thm: ISS outer-loop}
	The inexact outer loop is small-disturbance ISS. That is, for any $h>0$ and $\hat{K}_0 \in \mathcal{K}_h$, if $\norm{\tilde{K}} < f(h)$, there exist a $\mathcal{KL}$-function $\beta_1(\cdot,\cdot)$ and a $\mathcal{K}_\infty$-function $\gamma_1(\cdot)$ such that 
	\begin{align}
		\norm{{P}^p_{\hat{K}} - P^\star} \leq \beta_1(\norm{{P}^0_{\hat{K}} - P^*}, p) + \gamma_1(\norm{\tilde{K}}).
	\end{align}
\end{theorem}
\begin{proof}
	From Lemma \ref{lm:invariantOut}, $\hat{K}^p_K \in \mathcal{K}_h$ for any $p \in \mathbb{N}_+$. From \eqref{eq:PKPK'DiffTrace}, at the $p$'th iteration, we have
	\begin{align}\label{eq:PiDiffTrace}
		\begin{split}
			&\Tr({P}^p_{\hat{K}} - P^\star) \leq (1-\underline{f}_1(h)) \Tr({P}^{p-1}_{\hat{K}} - P^\star) \\
			&+ \bar{f}_2(h) \norm{R}\norm{\tilde{K}^{p}_K}^2.
		\end{split}
	\end{align}
	Repeating \eqref{eq:PiDiffTrace} for $p, p-1,\cdots,1$, 
	\begin{align}\label{eq:PiDiffTrace2}
			\Tr[{P}^p_{\hat{K}} - P^\star] \leq (1-\underline{f}_1)^{p} \Tr({P}_{\hat{K}}^{1} - P^\star) + \frac{\bar{f}_2 \norm{R}  \norm{\tilde{K}}_\infty^2}{\underline{f}_1(h)}.
	\end{align}
	It follows from \eqref{eq:PiDiffTrace2} and ~\cite[Theorem 2]{Mori1988} that 
	\begin{align}\label{eq:PiDiffTrace3}
			&\norm{P^p_{\hat{K}} - P^\star}_F \leq (1-\underline{f}_1)^{p} \sqrt{n}\norm{{P}^{1}_{\hat{K}} - P^\star}_F  + \frac{\bar{f}_2 \norm{R}\norm{\tilde{K}}_\infty^2}{\underline{f}_1}.
	\end{align}
	As $p\rightarrow \infty$, $P^p_{\hat{K}} \rightarrow P^\star$. The radius of the neighbor of $P^\star$ is proportional to $\norm{\tilde{K}}_\infty^2$. Thus, the proof follows.
\end{proof}


\subsubsection{Robustness of Maximizing Controller to Perturbations}
The perturbed inner-loop iteration \eqref{eq:inexact_iter_iloop} has inexact matrix $\hat{A}_{K,L}^{p,q}$, and sequences $\{\hat{L}_{q+1}(K_p)\}_{q=0}^{\infty}$, and $\{\hat{P}_{K,L}^{p,q}\}_{q=0}^{\infty}$.  We next analyze its robustness to perturbations when it differs from the exact loop matrices and sequences. 

\begin{lemma}[Stability of the Inner-Loop's System Matrix]
	Given $K \in \breve{\mathcal{K}}$, there exists a $g \in \reline_+$, such that if $\norm{\tilde{L}_{q+1}(K_p)}_F \leq g$,  $\hat{A}_{K,L}^{p,q}$ is Hurwitz for all $q \in \mathbb{N}_+$.
	\label{lm:ISS_InnerLoop}
\end{lemma}

\begin{lemma}
	For any  $h>0$ and $K \in \mathcal{K}_h$, let ${K}' = R^{-1}B^\top P_K$, where $P_K$ is the solution of \eqref{eq:brlemma2}, and $\hat{K}' = {K}' + \tilde{K}$. Then, there exists $f(h)>0$, such that $\hat{K}' \in \mathcal{K}_h$ as long as $\norm{\tilde{K}} < f(h)$.
	\label{lm:invariantOut}
\end{lemma}
\begin{theorem} 
	Assume $\norm{\tilde{L}_q(K_p)} < e$ for all $q \in \mathbb{N}_+$. There exists $\hat{\beta}(K) \in [0,1)$, and $\lambda(\cdot) \in \breve{\mathcal{K}}_\infty$, such that
	\begin{align}
		\norm{\hat{P}_{K,L}^{p,q}  - P_{K,L}^{p,q}}_F \leq \hat{\beta}^{q-1}(K) \Tr(P_{K,L}^{p,q} 
		) + \lambda(\norm{\tilde{L}}_\infty).
	\end{align}
	\label{thm:innerISS}
\end{theorem}
From Theorem \ref{thm:innerISS}, as $q \to \infty$, $\hat{P}_{K,L}^{p,q}$ approaches the solution $P_K$ and enters the ball centered at $P_{K,L}^{p,q}$ with radius proportional to $\norm{\tilde{L}}_\infty$. Hence, the proposed inner-loop iterative algorithm well approximates $P_{K,L}^{p,q}$.

	\section{Numerical Experiments}
\label{sec:sim}
We consider a humanoid robot model~\cite{Fierro2013,Pristovani_2018} in the form of  a three-link kinematic chain; and a standard double pendulum. The humanoid is non-minimum phase, underactuated, and possesses badly damped poles. Its passive joint can be modeled as a Wiener process noise that additively perturbs its dynamics. 

This model has three states: two upper hinge (the hip and knee) \textit{actuated joints}  and a lower hinge (the ankle) \textit{passive joint}. The dynamics is $x=[\theta_1, \theta_2, \theta_3, \dot{\theta}_1, \dot{\theta}_2, \dot{\theta}_3]^\top$, where $\theta_1$, $\theta_2$, and $\theta_3$ are the angles of the ankle, hip, and knee respectively. The linearized model of the triple inverted pendulum admits a form of the infinite dimensional linear PDE in \eqref{eq:po_leqg_opt}, where $A \in \mathbb{R}^{6 \times 6}$ and $B \in \mathbb{R}^{6 \times 2}$(see~\cite[Section 3]{FURUT1984}), and $D = \begin{bmatrix} 0_{3 \times 3}, I_{3}\end{bmatrix}^\top$. We impose an $\hinf$ norm bound of $\gamma = 5$ on the robot, set the initial state  to $x(0) = [0,-5,10,10,-10,10]^\top$ and set $C = \begin{bmatrix} I_6 , 0_{2 \times 6} \end{bmatrix}^\top, \quad E = \begin{bmatrix} 0_{6 \times 2} , I_2 \end{bmatrix}^\top$. Throughout, $w(t)$ is set to a Wiener process such that its time derivative  $dw$ is drawn from a zero-mean Gaussian distribution with variance $\mc{X}^2$. 
We chose a step size, $dt = 0.0001$. We next report our findings for the model-based, model-free algorithm, and the natural policy gradient algorithm (NPG)~\cite{ShamNPG}. For other numerical experiment reports, we refer readers to our recent conference paper~\cite{LekanIFAC}. 

\subsection{Model-based Mixed Design vs. NPG}
Let us describe numerical experiments on the algorithms described so far. At each iteration,  $\tilde{K}_p$ is sampled from a uniform Gaussian distribution whose Frobenius norm is $0.15$. We found 
 \begin{small}
	\begin{align*}
		\hat{K}_0 = \begin{bmatrix} -203.3  &-74.2 &-31.4 &-67.7  &-28.4  &-16.5 \\
			-529.5 &-198.8 &-77.8 &-175.5 &-78.7 &-39.0 \end{bmatrix}.
	\end{align*}
\end{small}
The results for running the model-based and NPD algorithms are shown in Figures \ref{fig:norm_NG_robust} and \ref{fig:norm_Nature_robust}. The robust mixed design PO scheme approaches the optimal solution after the $5$'th iteration (See \autoref{fig:norm_NG_robust}). At the last iteration, the deviation from the optimal cost matrix\footnote{Calculated as ${\norm{\hat{P}^{20}_{K}- P^\star}}/{\norm{P^\star}_F}$.} is $2.9\%$, while the gain error\footnote{Calculated as ${\norm{\hat{K}^{20}_{K} - K^\star}_F}/{\norm{K^\star}_F}$.} 
is $2.6\%$. In contrast, NPG exhibits cost matrix and controller gain errors that are unbounded as the iteration lengthens. 

We compared the time it takes to compute the optimal policies in model-based nested algorithm against NPG in Table \ref{tab:compTime}. For the double and triple inverted pendulums, the computational time of our algorithm is much less than that of NPG by around $90\%$. This is a validation of our superior convergence rate compared to NPG's sublinear convergence rate. 
\begin{table}[tb]
\centering
\caption{Computational Time: Model-based PO vs. NPG.}
\begin{tabular}{ |p{1.5cm}|p{1.5cm}|p{1.3cm}||p{1.5cm}|p{1.3cm}|p{1.3cm}|  }
	\hline
	\multicolumn{6}{|c|}{Policy Optimization Computational time (secs)} \\
	\hline
	\multicolumn{3}{|c||}{Double Inverted Pendulum} & \multicolumn{3}{|c|}{Triple Inverted Pendulum} \\
	\hline
	Model-based & Model-free & NPG  & Model-based & Model-free &NPG\\
	\hline
	$0.0901$ & $0.3061$   & $2.1649$   & $0.1455$  & $0.7829$  &   $2.3209$\\
	\hline
\end{tabular}
\label{tab:compTime}
\end{table}
\begin{figure}[t!]
	\centering
	\begin{subfigure}{.32\linewidth}
		\centering
		\includegraphics[width=0.99\linewidth]{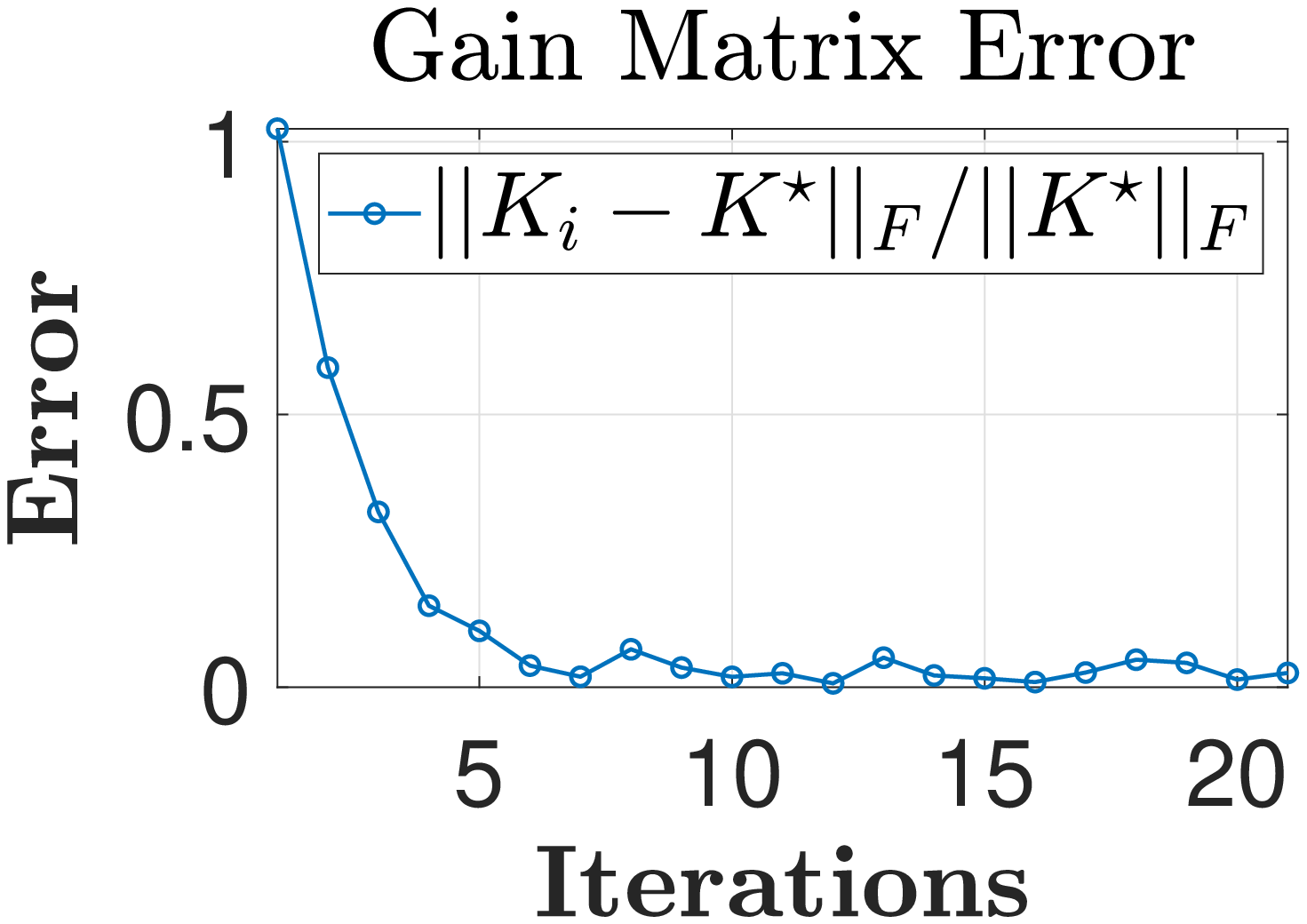}
		\label{fig:Knorm_NG_robust}
	\end{subfigure}
	\begin{subfigure}{.32\linewidth}
		\centering
		\includegraphics[width=0.99\linewidth]{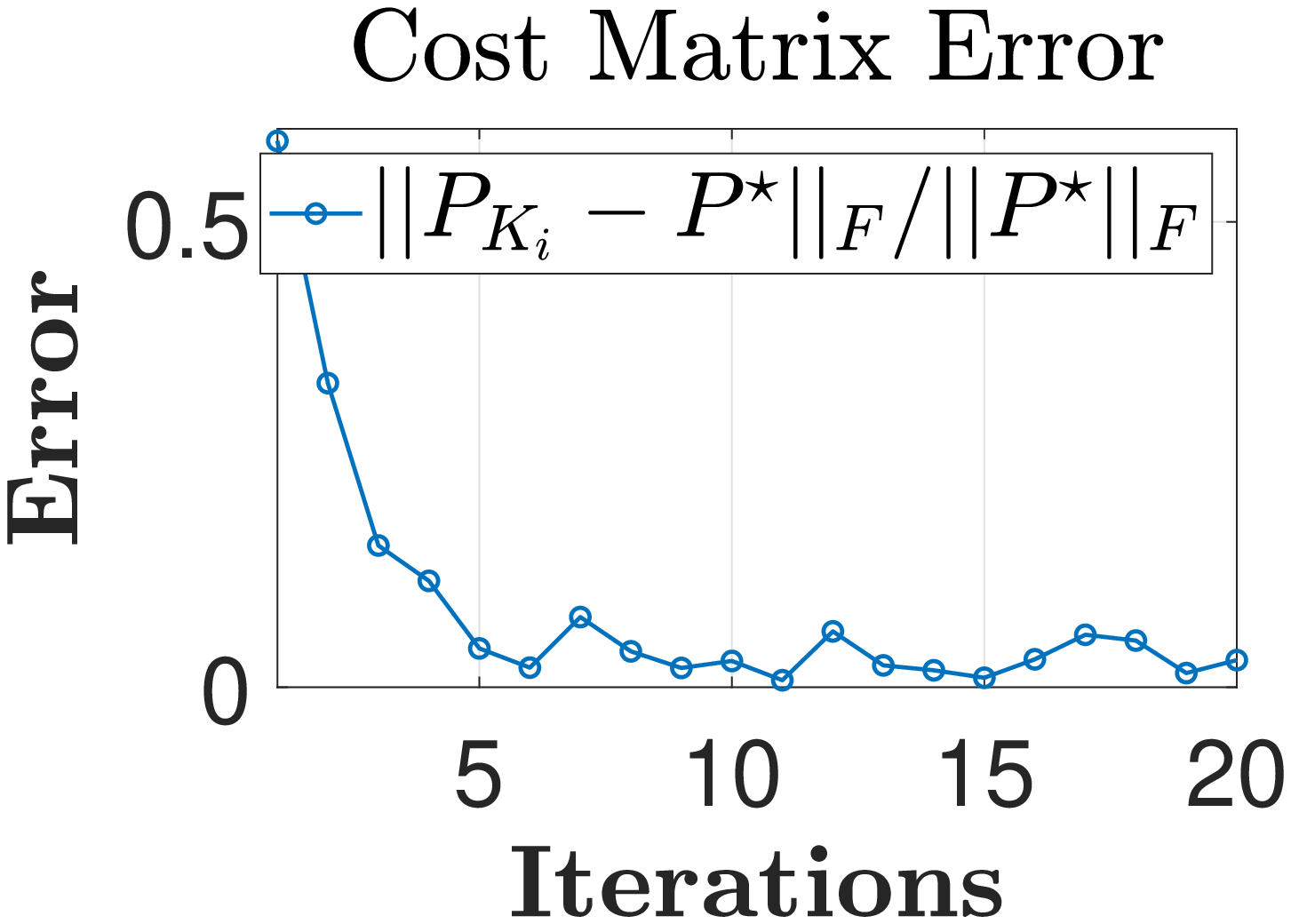}
		\label{fig:Pnorm_NG_robust}
	\end{subfigure}
	\begin{subfigure}{.32\linewidth}
		\centering
		\includegraphics[width=0.99\linewidth]{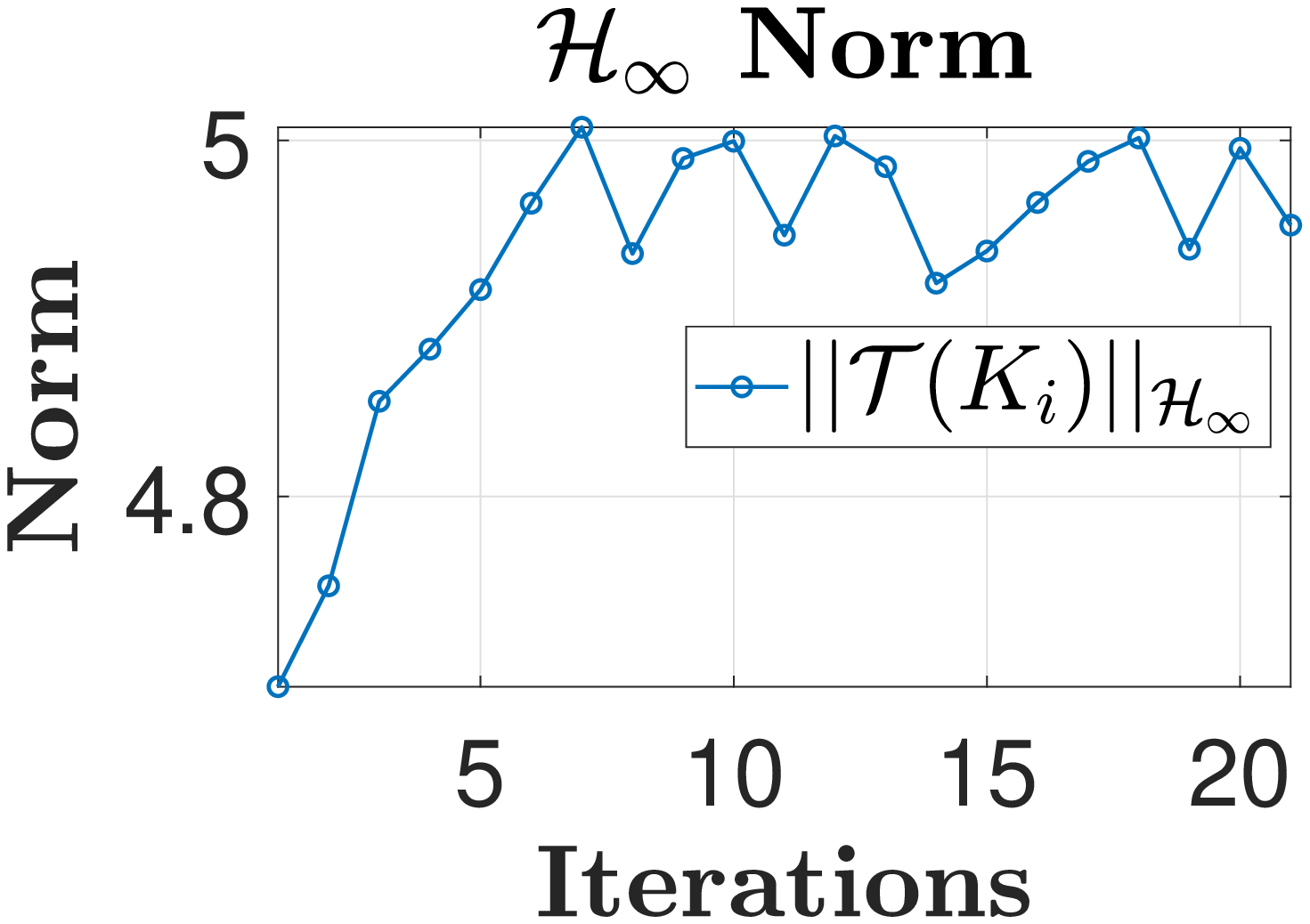}
		\label{fig:Hnorm_NG_robust}
	\end{subfigure}
	\caption{Model-based design: $\norm{\tilde{K}}_\infty = 0.15$.}
\label{fig:norm_NG_robust}
\end{figure}
%
%
\begin{figure}[t!]
	\centering
	\begin{subfigure}{.32\linewidth}
		\includegraphics [width=0.99\linewidth]{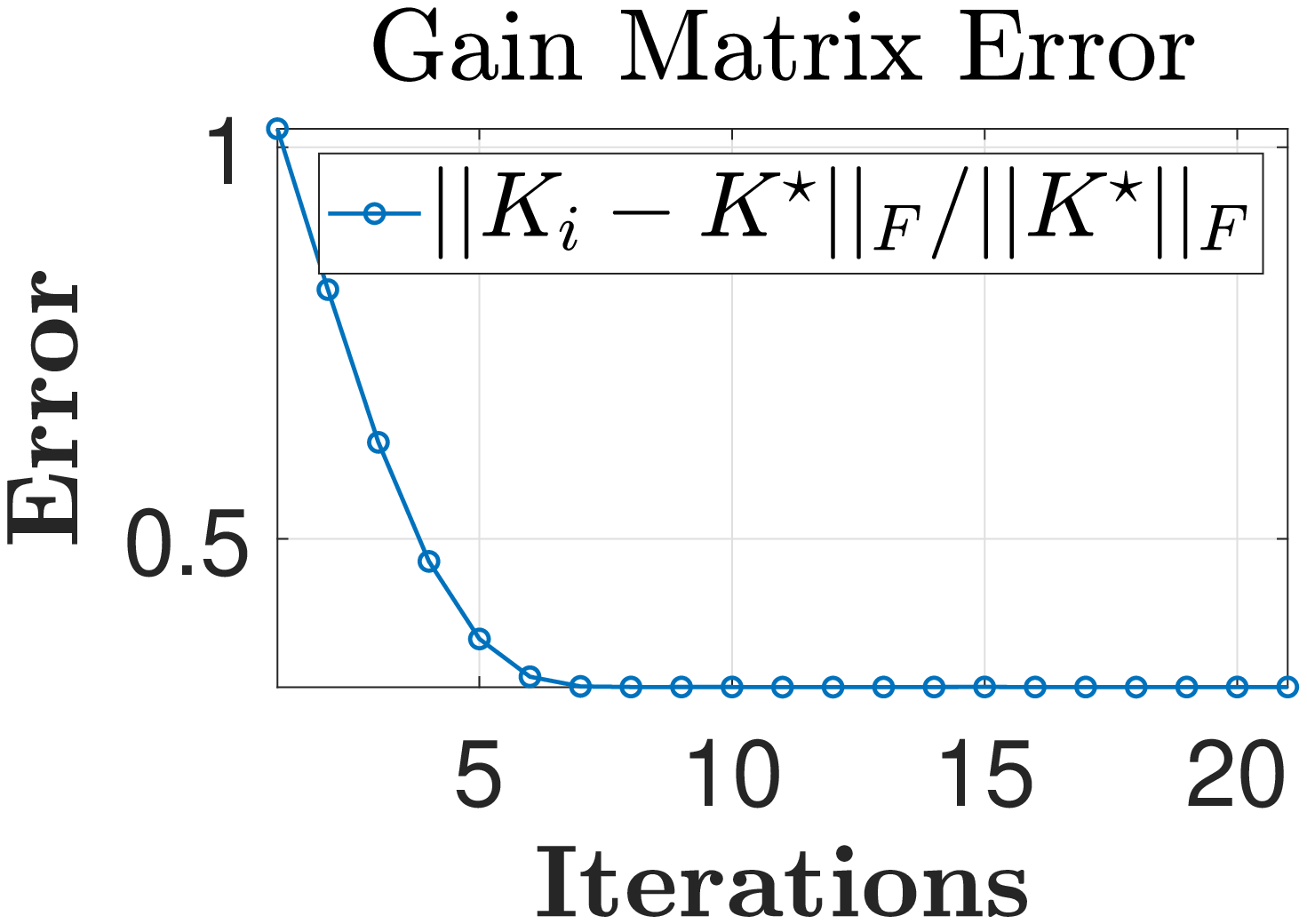}
		\label{fig: K evolution}  
	\end{subfigure}
	\begin{subfigure}{.32\linewidth}
		\centering
		\includegraphics [width=0.99\linewidth]{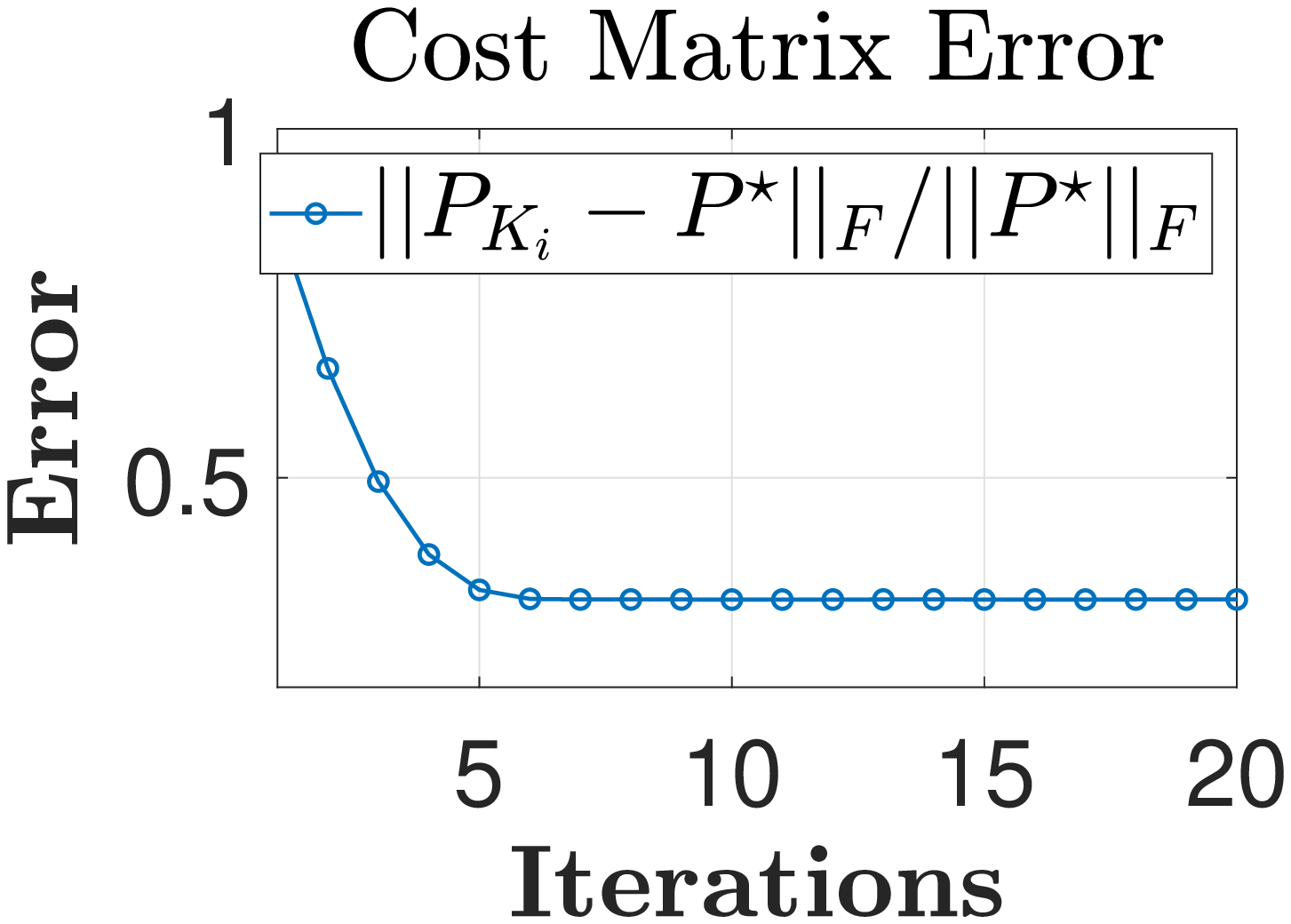}
		\label{fig: P evolution}
	\end{subfigure}	
	\begin{subfigure}{.32\linewidth}
		\centering
		\includegraphics [width=0.99\linewidth]{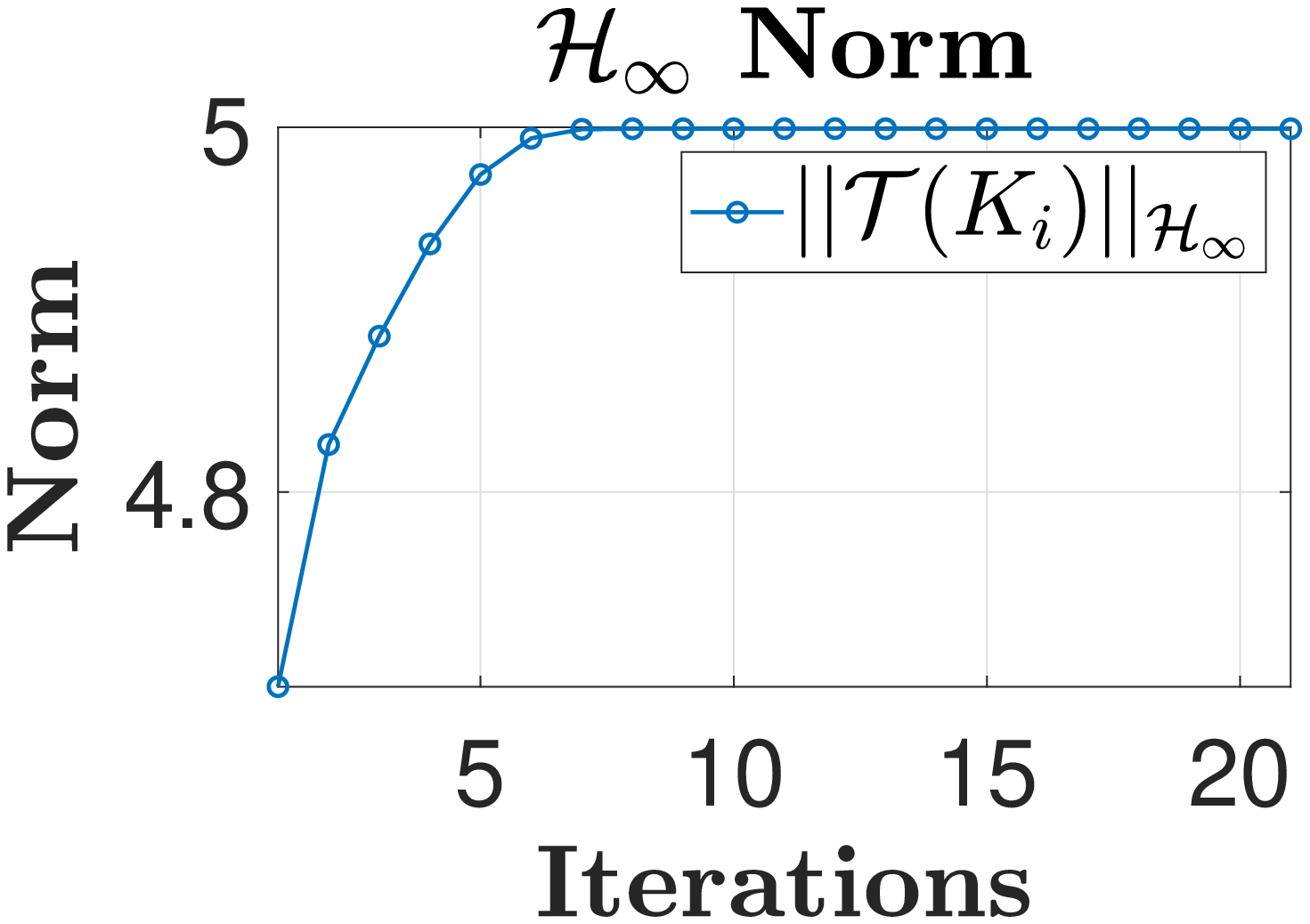}
		\label{fig: H evolution}
	\end{subfigure}	
	\caption{Sampling-based scheme results.}
	\label{fig:Learn_triple}
\end{figure}
\begin{figure}[tb!]
\centering
\begin{subfigure}{.32\linewidth}
	\centering
	\includegraphics[width=0.99\linewidth]{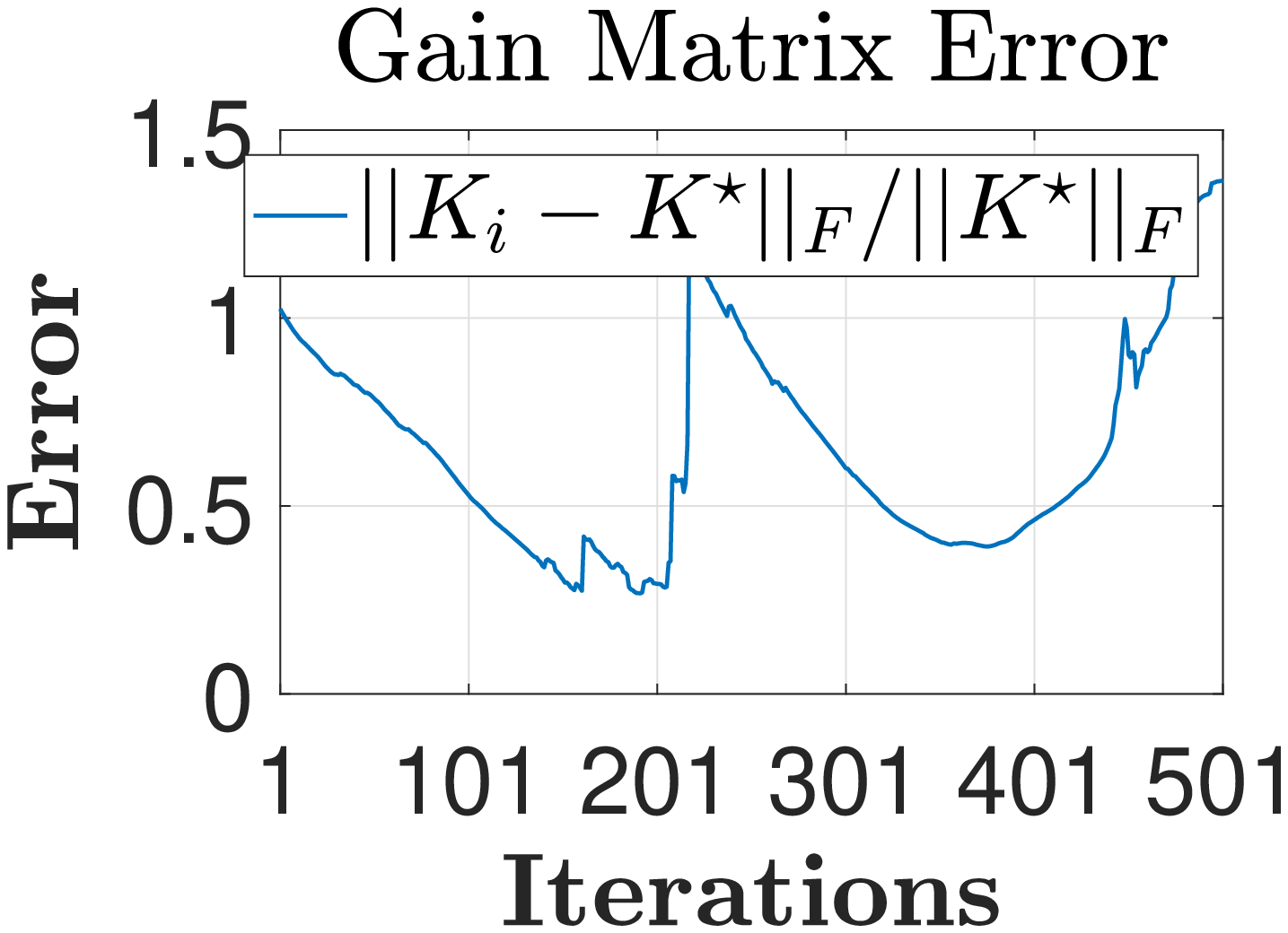}
	\label{fig:Knorm_Nature_robust}
\end{subfigure}
\begin{subfigure}{.32\linewidth}
	\centering
	\includegraphics[width=0.99\linewidth]{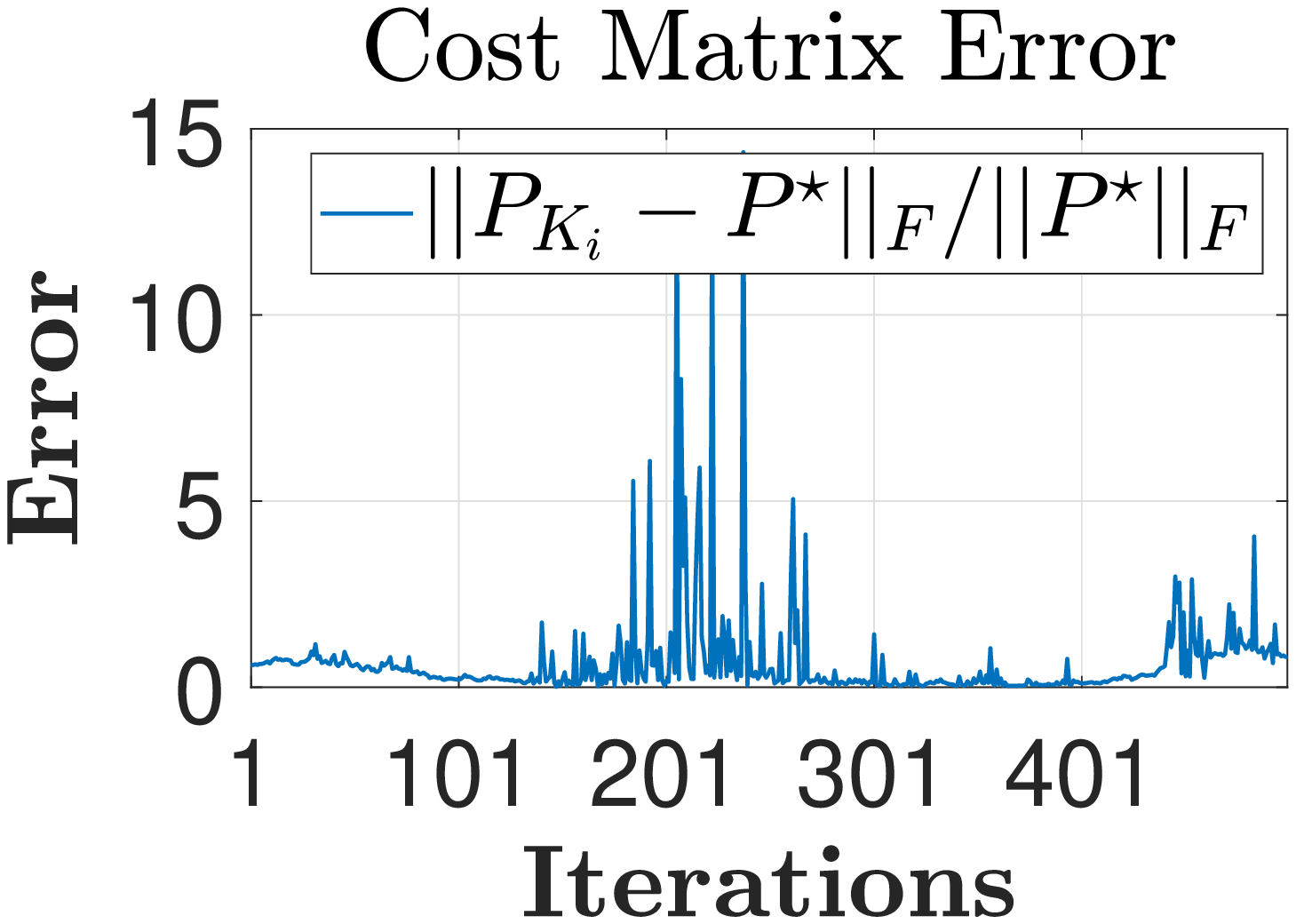}
	\label{fig:Pnorm_Nature_robust}
\end{subfigure}
\begin{subfigure}{.32\linewidth}
	\centering
	\includegraphics[width=0.99\linewidth]{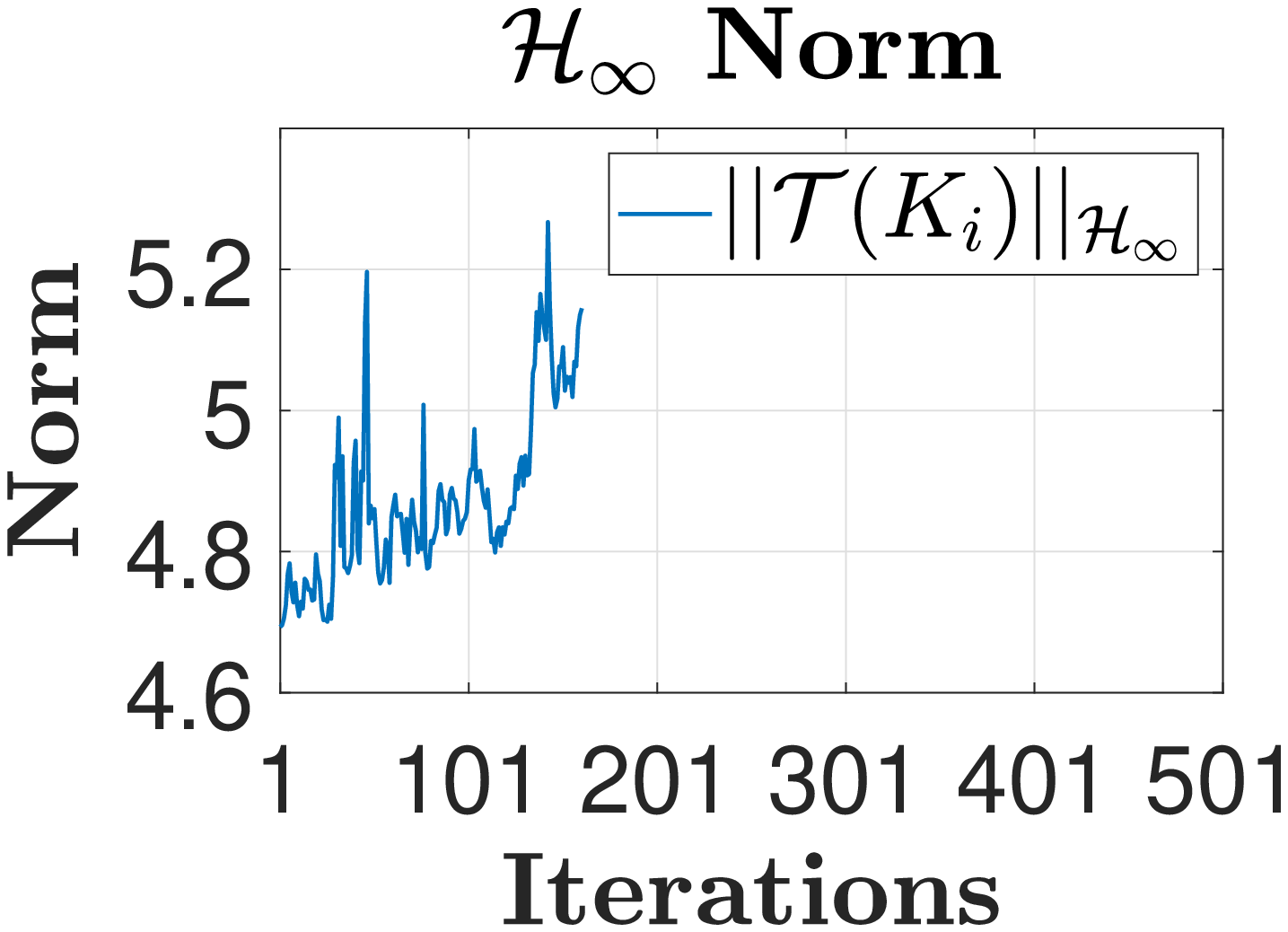}
	\label{fig:Hnorm_Nature_robust}
\end{subfigure}
\caption{Model-based design vs. NPG with $\norm{\tilde{K}}_\infty = 0.1$.}
\label{fig:norm_Nature_robust}
\end{figure}
\subsection{Sampling-based Mixed Design vs. NPG}
For the model-based algorithm, we set $\bar{p} = 20$ and found the maximum data collection time before attaining the full rank condition of Lemma \ref{lm:rank_cond} to be $t_l = 1500s$. The system parameters $A$ and $B$ are unknown but the initial controller $\hat{K}_0 \in \mathcal{K}$ is searched for following~\cite[Alg. 1]{LekanIFAC}. 

We run the sampling-based algorithm on \eqref{eq:po_leqg_opt}. From the charts of Fig. \ref{fig:Learn_triple}, the controller $\hat{K}_p$ found at each iteration converges after $5$ iterations alongside $\hat{P}_{K}^p$. At the $20$'th iteration, the relative error ${||\hat{K}_{20}-K_{*}||}/{||K_{*}||} = 31.5\%$ and ${||\hat{P}_{K_{20}}-P_{*}||}/{||P_{*}||} = 31.6\%$. These demonstrate that the proposed algorithm does find an approximate optimal solution from the noisy data.

Comparing our algorithm under an additive Wiener process noise against the natural policy gradient, we find that the relative errors in gain and cost matrix errors are not well-behaved. The same disruption applies to the $\hinf$-norm plot. These further demonstrates the need for a robust PO scheme such as the one we have presented in this work for problems that fall into the class of systems \eqref{eq:po_leqg_opt}.
	\ifCLASSOPTIONcaptionsoff
	\newpage
	\fi
	\bibliographystyle{ieeetr}        
	\bibliography{mixednote} 

\begin{thebibliography}{10}

\bibitem{Fazel2018}
M.~Fazel, R.~Ge, S.~Kakade, and M.~Mesbahi, ``Global convergence of policy
  gradient methods for the linear quadratic regulator,'' in {\em Proceedings of
  the 35th International Conference on Machine Learning}, vol.~80 of {\em
  Proceedings of Machine Learning Research}, pp.~1467--1476, PMLR, 10--15 Jul
  2018.

\bibitem{Mohammadi2022}
H.~Mohammadi, A.~Zare, M.~Soltanolkotabi, and M.~R. Jovanović, ``Convergence
  and sample complexity of gradient methods for the model-free
  linear–quadratic regulator problem,'' {\em IEEE Transactions on Automatic
  Control}, vol.~67, no.~5, pp.~2435--2450, 2022.

\bibitem{HuAnnualRevs}
B.~Hu, K.~Zhang, N.~Li, M.~Mesbahi, M.~Fazel, and T.~Ba{\c{s}}ar, ``Toward a
  theoretical foundation of policy optimization for learning control
  policies,'' {\em Annual Review of Control, Robotics, and Autonomous Systems},
  vol.~6, pp.~123--158, 2023.

\bibitem{Gravell2021}
B.~Gravell, P.~M. Esfahani, and T.~Summers, ``Learning optimal controllers for
  linear systems with multiplicative noise via policy gradient,'' {\em IEEE
  Transactions on Automatic Control}, vol.~66, no.~11, pp.~5283--5298, 2021.

\bibitem{zhang2021derivative}
K.~Zhang, X.~Zhang, B.~Hu, and T.~Basar, ``Derivative-free policy optimization
  for linear risk-sensitive and robust control design: Implicit regularization
  and sample complexity,'' {\em Advances in Neural Information Processing
  Systems}, vol.~34, pp.~2949--2964, 2021.

\bibitem{Zhang2020SIAM}
K.~Zhang, A.~Koppel, H.~Zhu, and T.~Ba\c{s}ar, ``Global convergence of policy
  gradient methods to (almost) locally optimal policies,'' {\em SIAM Journal on
  Control and Optimization}, vol.~58, no.~6, pp.~3586--3612, 2020.

\bibitem{Zhang2021}
K.~{Zhang}, B.~{Hu}, and T.~{Ba{\c{s}}ar}, ``{Policy Optimization for
  $\mathcal{H}_2$ Linear Control with $\mathcal{H}_\infty$ Robustness
  Guarantee: Implicit Regularization and Global Convergence},'' {\em arXiv
  e-prints}, p.~arXiv:1910.09496, oct 2019.

\bibitem{LevineEnd2End}
S.~Levine, C.~Finn, T.~Darrell, and P.~Abbeel, ``{End-to-End Training of Deep
  Visuomotor Policies},'' {\em The Journal of Machine Learning Research},
  vol.~17, no.~1, pp.~1334--1373, 2016.

\bibitem{RechtTour}
B.~Recht, ``A tour of reinforcement learning: The view from continuous
  control,'' {\em Annual Review of Control, Robotics, and Autonomous Systems},
  vol.~2, pp.~253--279, 2019.

\bibitem{Pang2021}
B.~{Pang} and Z.~P. {Jiang}, ``Adaptive optimal control of linear periodic
  systems: an off-policy value iteration approach,'' {\em IEEE Transactions on
  Automatic Control}, vol.~66, no.~2, pp.~888--894, 2021.

\bibitem{DeanSampleComplexity}
S.~Dean, H.~Mania, N.~Matni, B.~Recht, and S.~Tu, ``On the sample complexity of
  the linear quadratic regulator,'' {\em Foundations of Computational
  Mathematics}, vol.~20, no.~4, pp.~633--679, 2020.

\bibitem{Glover1989}
K.~Glover, ``Minimum entropy and risk-sensitive control: the continuous time
  case,'' in {\em Proceedings of the 28th IEEE Conference on Decision and
  Control,}, pp.~388--391 vol.1, 1989.

\bibitem{Khargonekar1988}
P.~Khargonekar, I.~Petersen, and M.~Rotea, ``$\mathcal{H}_\infty$ optimal
  control with state-feedback,'' {\em IEEE Transactions on Automatic Control},
  vol.~33, no.~8, pp.~786--788, 1988.

\bibitem{basar1990minimax}
T.~Basar, ``Minimax disturbance attenuation in ltv plants in discrete time,''
  in {\em 1990 American Control Conference}, pp.~3112--3113, IEEE, 1990.

\bibitem{SteeleStochCalc}
J.~M. Steele, {\em Stochastic calculus and financial applications}, vol.~1.
\newblock Springer, 2001.

\bibitem{Oksendal}
B.~{\O}ksendal and B.~{\O}ksendal, {\em Stochastic differential equations}.
\newblock Springer, 2003.

\bibitem{ShamNPG}
S.~M. Kakade, ``A natural policy gradient,'' {\em Advances in neural
  information processing systems}, vol.~14, 2001.

\bibitem{Zhang2019}
K.~Zhang, Z.~Yang, and T.~Basar, ``Policy optimization provably converges to
  nash equilibria in zero-sum linear quadratic games,'' in {\em Advances in
  Neural Information Processing Systems} (H.~Wallach, H.~Larochelle,
  A.~Beygelzimer, F.~d\textquotesingle Alch\'{e}-Buc, E.~Fox, and R.~Garnett,
  eds.), vol.~32, Curran Associates, Inc., 2019.

\bibitem{Bu2019}
J.~{Bu}, L.~J. {Ratliff}, and M.~{Mesbahi}, ``{Global Convergence of Policy
  Gradient for Sequential Zero-Sum Linear Quadratic Dynamic Games},'' {\em
  arXiv e-prints}, Nov. 2019.

\bibitem{Duncan2013}
T.~E. Duncan, ``{Linear-Exponential-Quadratic Gaussian} control,'' {\em IEEE
  Transactions on Automatic Control}, vol.~58, no.~11, pp.~2910--2911, 2013.

\bibitem{Kantarovich}
{\em {Functional Analysis in Normed Spaces}}.
\newblock New York: MacMillan, 1964.

\bibitem{book_Basar}
T.~Başar and P.~Bernhard, {\em $H_\infty$-Optimal Control and Related Minimax
  Design Problems: A Dynamic Game Approach}.
\newblock Springer, 2008.

\bibitem{Kleinman1968}
D.~Z. Kleinman, ``On an iterative technique for riccati equation
  computations,'' {\em IEEE Transactions on Automatic Control}, vol.~13,
  pp.~114--115, 1968.

\bibitem{LekanIFAC}
L.~Molu, ``Mixed $\mathcal{H}_2/\mathcal{H}_\infty$ policy synthesis.,'' in
  {\em The International Federation of Automatic Control, 22nd World Congress},
  p.~arXiv:2302.08846, July 2023.

\bibitem{boyd1994linear}
S.~Boyd, L.~El~Ghaoui, E.~Feron, and V.~Balakrishnan, {\em Linear matrix
  inequalities in system and control theory}.
\newblock SIAM, 1994.

\bibitem{Sontag2008}
E.~D. Sontag, {\em Input to State Stability: Basic Concepts and Results},
  pp.~163--220.
\newblock Berlin, Heidelberg: Springer Berlin Heidelberg, 2008.

\bibitem{DuncanSDEBrownian}
T.~E. Duncan, B.~Maslowski, and B.~Pasik-Duncan, ``Control of some linear
  stochastic systems in a hilbert space with fractional brownian motions,'' in
  {\em 2011 16th International Conference on Methods \& Models in Automation \&
  Robotics}, pp.~107--110, IEEE, 2011.

\bibitem{DuncanStochastic}
T.~E. Duncan and B.~Pasik-Duncan, ``Stochastic linear-quadratic control for
  systems with a fractional brownian motion,'' in {\em 49th IEEE Conference on
  Decision and Control (CDC)}, pp.~6163--6168, IEEE, 2010.

\bibitem{JiangJiang}
Y.~Jiang and Z.-P. Jiang, ``{Computational Adaptive Optimal Control for
  Continuolus-Time Linear Systems With COmpletely Unknown Dynamics},'' vol.~48,
  pp.~2699--2704, 2023.

\bibitem{Mori1988}
T.~Mori, ``Comments on "a matrix inequality associated with bounds on solutions
  of algebraic {Riccati} and {Lyapunov} equation" by {J. M. Saniuk} and {I.B.
  Rhodes},'' {\em IEEE Transactions on Automatic Control}, vol.~33, no.~11,
  pp.~1088--, 1988.

\bibitem{Fierro2013}
M.~González-Fierro, C.~Balaguer, N.~Swann, and T.~Nanayakkara, ``A humanoid
  robot standing up through learning from demonstration using a multimodal
  reward function,'' in {\em 2013 13th IEEE-RAS International Conference on
  Humanoid Robots (Humanoids)}, pp.~74--79, 2013.

\bibitem{Pristovani_2018}
R.~D. Pristovani, D.~R. Sanggar, and P.~Dadet, ``Implementation of push
  recovery strategy using triple linear inverted pendulum model in
  {\textquotedblleft}t-{FloW}{\textquotedblright} humanoid robot,'' {\em
  Journal of Physics: Conference Series}, vol.~1007, p.~012068, apr 2018.

\bibitem{FURUT1984}
K.~Furut, T.~Ochiai, and N.~Ono, ``Attitude control of a triple inverted
  pendulum,'' {\em International Journal of Control}, vol.~39, no.~6,
  pp.~1351--1365, 1984.

\bibitem{MAGNUS1985}
J.~R. Magnus and H.~Neudecker, ``Matrix differential calculus with applications
  to simple, hadamard, and kronecker products,'' {\em Journal of Mathematical
  Psychology}, vol.~29, no.~4, pp.~474--492, 1985.

\bibitem{book_horn}
R.~A. Horn and C.~R. Johnson, {\em Matrix Analysis, second edition}.
\newblock Cambridge University Press, 2013.

\bibitem{Zhou_Robust}
K.~Zhou, J.~C. Doyle, and K.~Glover, {\em Robust and Optimal Control}.
\newblock Prentice hall Upper Saddle River, NJ, 1996.

\end{thebibliography}
	\newcounter{appidx}
\setcounter{appidx}{1}
\renewcommand\theequation{\Alph{appidx}.\arabic{equation}}

\setcounter{equation}{0}

\section*{Appendix \Alph{appidx}: Lemmas and Proofs}
In this appendix, we introduce a series of lemmas to guide our problem description and proposed solution. 

\begin{proof}[Proof of Lemma \ref{lm:oloop_iter}]
	When $p = 0$, $K_0 \in \mathcal{K}$, and it satisfies (1) (See \cite[Alg. 1]{LekanIFAC}.)  
	For $p>0$, introduce the identities,
	\begin{subequations}
		\begin{align}
			RK_{p+1} &= B^\top P_K^p, \qquad K_{p+1}^\top R = P_K^p B, \\
			A_K^{p^\top}P_K^p &= 	A_K^{(p+1)^\top}P_K^p + (K_{p+1} - K_p)^\top B^\top P_K^p, \\
			P_K^pA_K^{p} &= 	P_K^p A_K^{(p+1)} + P_K^p B (K_{p+1} - K_p).
		\end{align}
		\label{eq:riccati_identities}
	\end{subequations}
	Therefore,  equation \eqref{eq:riccati_outer_loop_iter} becomes
	%
	%
	\begin{align}
		&A_K^{(p+1)^\top}P_K^p + P_K^p A_K^{(p+1)} + \gamma^{-2} P_K^p DD^\top P_K^p +C^\top C   
		\label{eq:riccati_outer_converge_lm1} \\ 
		&  + K_{p+1}^\top R K_{p+1}  +(K_{p+1} - K_{p})^\top R (K_{p+1} -   K_p) = 0. \nonumber
	\end{align}
	Thus, for a stabilizing $K_{p+1} (\neq K_p)$ we must have $(K_{p+1} - K_{p})^\top R (K_{p+1} -   K_p) \succ 0$  so that
	\begin{align}
		A_K^{(p+1)^\top}P_K^p + P_K^p A_K^{(p+1)} + \gamma^{-2} P_K^p DD^\top P_K^p + Q_K^{p+1}  \prec 0.
		\label{eq:lemma2_st1}
	\end{align} 
	If (read: since) the inequality \eqref{eq:lemma2_st1} holds, the bounded real Lemma~\cite[Lemma A.1, statement 3]{Zhang2021} stipulates that a $P_K^p \succ 0$ exists; by~\cite[Lemma A.1, statement 1]{Zhang2021}, $\|T_{zw}(K_p)\|_\infty < \gamma$ given that $\lambda_i(A_K^{(p+1)}) <0$ in \eqref{eq:lemma2_st1}. \textit{A fortiori}, $K_p \in \mc{K}$ for $p >0$ by the bounded real Lemma. This proves the first statement.
	
	The proof for statement (2) now follows. At the $(p+1)$'th iteration, it can be verified that \eqref{eq:riccati_outer_loop_iter} admits the form
	\begin{align}
		&A_K^{(p+1)^\top}P_K^{p+1} + P_K^{p+1} A_K^{(p+1)} + C^\top C + K_{p+1}^\top R K_{p+1}\nonumber \\
		& \qquad \qquad  + \gamma^{-2} P_K^{p+1} DD^\top P_K^{p+1} = 0,
		\label{eq:riccati_outer_converge_lm1.2}
	\end{align}
	so that subtracting \eqref{eq:riccati_outer_converge_lm1.2} from \eqref{eq:riccati_outer_converge_lm1} (at the $p$'th iteration) and using the statistical independence property of the noise term $w(t)$ (from Ass. \ref{ass:realizability}) \ie $DD^\top = 0$, we have
	\begin{align}
		&A_K^{(p+1)^\top}\left[P_K^{p} - P_K^{p+1}\right]+ \left[P_K^{p} - P_K^{p+1}\right]A_K^{(p+1)} \nonumber \\
		&\qquad  (K_{p+1}- K_p)^\top R (K_{p+1}-K_p)=0. 
		\label{eq:riccati_outer_converge_lm1.4}
	\end{align}
	\textbf{Observe}: 
	Equation \eqref{eq:riccati_outer_converge_lm1.4} is  a Lyapunov equation of the form 
	\begin{align}
		&A_K^{(p+1)^\top}\left[P_K^{p} - P_K^{p+1}\right]+ \left[P_K^{p} - P_K^{p+1}\right]A_K^{(p+1)} \nonumber \\
		&\qquad =- (K_{p+1}- K_p)^\top R (K_{p+1}-K_p) \triangleq \Psi_p
	\end{align}
	Statement 1 of Lemma \ref{lm:inverseLya} implies that  $A_K^{(p+1)}$ is Hurwitz since $\Psi_p$ above is Hurwitz. Hence, we must have $P_K^{p} - P_K^{p+1} \succeq 0$ because $(K_{p+1}- K_p)^\top R (K_{p+1}-K) \succeq 0$ by statement (3) of Lemma \ref{lm:ZhouRobust}.  
	Whence, $P_K^p \succeq P_K^{p+1}$ and $K_{p+1} \ge K_p$. This proves the second statement. In this sentiment, the sequence $\{P_K^{p}\}_{p=0}^{\infty}$ is decreasing, bounded below by $0$ and has a finite norm so that $\{P_K^{p}\}_{p=0}^{\infty}$ converges to $P_{K}^\infty$. This satisfies \eqref{eq:GARE}. Observe from equation \eqref{eq:riccati_outer_loop_iter} that $P_K^p$ is self-adjoint so that from the ``limit of monotonic positive operators theorem"~\cite[p. 189]{Kantarovich},  $\lim_{p \rightarrow \infty}\norm{P_{K}^p - P^*}_F = 0$. By a similar argument for decreasing operators sequences~\cite[p. 190]{Kantarovich}, the sequence $\{K_K^{p}\}^{\infty}_{p=0}$ is increasing and upper bounded by $K_K^\infty$. Hence, $\lim_{p \rightarrow \infty}\norm{K_p - K^*}_F = 0$. The third statement is thus proven.
	\label{proof:lm_oloop_iter}
\end{proof}
\begin{proof}[Proof of Lemma \ref{lm:EK=0}]
	Since $R \succ 0$, $\Psi_K = 0$ implies $K = K^\prime$. Therefore at $\Psi_K =0$, we must have $K = K^\prime$ which implies that $P_K = P_K^\prime$. If  $K = K^\prime$ and $P_K = P_K^\prime$, it suffices to conclude that $K^\prime= K \triangleq K^\star$ where $K^\star = R^{-1}B^\top P^\star $. Hence,
	$\Psi_K =0$ is tantamount to  $P_K = P^\star$ and $K = K^\star$. 
\end{proof}

\begin{proof}[Proof of Lemma \ref{lm:bound_EK}]
	Define $A^\star = A - BR^{-1}B^\top P^\star + \gamma^{-2}DD^\top P^\star$ so that \eqref{eq:brlemma2} becomes
	\begin{align}\label{eq:AREAoptPK}
		&A^{\star \top}P_{K}^p + P_{K}^pA^\star + Q_{K_p}  + (K^\star - K_p)^\top R K_p^\prime \nonumber  \\
		&+ K_p^{\prime\top} R(K^\star - K_p)  
		-\gamma^{-2} P^\star DD^\top P_{K}^p  - \gamma^{-2}P_{K}^pDD^\top P^\star \nonumber \\
		& + \gamma^{-2}P_{K}^p DD^\top P_{K}^p = 0.
	\end{align}
	In addition, \eqref{eq:GARE} can be rewritten (replacing $\alpha$ with $\gamma$) as
	\begin{align}
			&A^{\star \top} P^\star + P^\star A^\star + Q + K^{\star \top}RK^\star - \gamma^{-2}P^\star DD^\top P^\star = 0. 
			\label{eq:GAREAoptPopt}
	\end{align}
	Subtracting \eqref{eq:GAREAoptPopt} from \eqref{eq:AREAoptPK} and completing squares, we have
	\begin{align} \label{eq:Aopt(PK-Popt)}
		\begin{split}
			&A^{\star \top} (P_{K}^p - P^\star) + (P_{K}^p - P^\star)A^\star + \Psi_{K_p} \\
			&+ \gamma^{-2}(P_{K}^p - P^\star)DD^\top(P_{K}^p - P^\star) \\
			& - (K_p^\prime- K^\star)^\top R(K_p^\prime- K^\star) = 0.  
		\end{split}
	\end{align}
	Let $\tilde{P}_K^p := P_K^p - P^\star$. It follows from $K_p^\star = R^{-1}B^\top P^\star$ and \eqref{eq:Aopt(PK-Popt)} that 
	\begin{align}\label{eq:Aopt(PK-Popt)3}
		\begin{split}
			&A^{\star \top} \tilde{P}_{K}^p + \tilde{P}_{K}^p(A-BR^{-1}B^\top P_K^p + \gamma^{-2}DD^\top P_K^p)  \\
			&\qquad \qquad + \Psi_{K_p}= 0,
		\end{split}
	\end{align}    
	whereupon $\mathcal{A}(K_p) \vec(\tilde{P}_K^p) = -\vec(\Psi_K)$ %
	with $\mathcal{A}(K_p)$ being
	\begin{align}
		\begin{split}
			&I_n \otimes A^{\star \top}   + (A-BR^{-1}B^\top P_K^p + \gamma^{-2}DD^\top P_K^p)^\top \otimes I_n.    
		\end{split}
	\end{align}

	From \eqref{eq:brlemma2} and the implicit function theorem,  $P_K^p$ is a continuously differentiable function of $K_p \in \mathcal{K}$. Since $A^\star$ is Hurwitz, there exists a ball $\mathcal{B}(K^\star, \delta) :=\{K \in \mathcal{K}|\norm{K - K^\star}_F \leq \delta \} $, such that $\mathcal{A}(K)$ is invertible for any $K \in \mathcal{K}_h \cap \mathcal{B}(K^\star, \delta)$. Therefore, for any $K \in \mathcal{K}_h \cap \mathcal{B}(K^\star, \delta)$, it follows that 
	\begin{align}\label{eq:DeltaPKNorm1}
		\norm{\tilde{P}_K^p}_F \leq \sgmin^{-1}(\mathcal{A}(K_p))\norm{\Psi_{K_p}}_F.
	\end{align}

	On the other hand, for any $K \in \mathcal{K}_h \cap \mathcal{B}^c(K^\star,\delta)$, where $\mathcal{B}^c$ is a complement of $\mathcal{B}$, $\Psi_{K_p} \neq 0$ and there exists a constant $b_1>0$ such that $\norm{\Psi_{K_p}} \geq b_1$. Thus, by Lemma \ref{lm:normTrace}, we have
	\begin{align}
		\norm{\tilde{P}_K^p}_F \leq \Tr(P_K^p) \leq \frac{h + \Tr(P^\star)}{b_1} \norm{\Psi_{K_p}}_F.
	\end{align}
	Suppose that  $b_2= \max_{K \in \mathcal{K}_h \cap \mathcal{B}_\delta(K^\star)}\sgmin^{-1}(\mathcal{A}(K))$ and $b(h)= \max \{b_2, \frac{h + \Tr(P^\star)}{b_1} \}$, 
	then the proof follows from \eqref{eq:DeltaPKNorm1} and the foregoing.
\end{proof}

\begin{proof}[Proof of Lemma \ref{lm:inner_loop_iter}]
	To prove the first statement, we proceed by induction. 
	For a $p\ge 0$ we have  $K_p \in \breve{\mathcal{K}}$ by Theorem \ref{thm:oloop_converge_rate}. 
	Subtracting \eqref{eq:riccati_inner_loop_iter} from \eqref{eq:riccati_outer_loop_iter} yields
	\begin{align}
		&0=A_K^{p\top} (P_K^p - P_{K,L}^{p,q}) + (P_K^p - P_{K,L}^{p,q})A_K^p +\nonumber \\
		& \gamma^2\left[L_{q+1}(K_p) - L_q(K_p)\right]^\top\left[L_{q+1}(K_p) - L_q(K_p)\right].
		\label{eq:P_KL_Diff}
	\end{align}
	In equation \eqref{eq:P_KL_Diff}, we have that  $\left[L_{q+1}(K_p) - L_q(K_p)\right]^\top\left[L_{q+1}(K_p) - L_q(K_p)\right] \succeq 0$ so that \eqref{eq:P_KL_Diff} admits a Lyapunov equation form. Following statement 2 of Lemma \ref{lm:ZhouRobust}, 
	we must have $(P_K^p - P_{K,L}^{p,q}) \succeq 0$. \textit{A fortiori}, we must have $A_{K,L}^{(p,q)}$ as  Hurwitz in \eqref{eq:P_KL_Diff} following statement 2 of Lemma \ref{lm:ZhouRobust}. This proves the first statement.
	
	To prove the second statement, we abuse notation by dropping the templated argument in $L_q(K_p)$. Let us consider the identities,
	\begin{align}
		&A_{K,L}^{(p,q)^\top} P_{K,L}^{p,q} = A_{K,L}^{(p,q+1)^\top} P_{K,L}^{p,q} - \gamma^2 \left[L_{q+1}-L_q\right]^\top L_{q+1} \nonumber \\
		&P_{K,L}^{p,q}A_{K,L}^{(p,q)} =  P_{K,L}^{p,q} A_{K,L}^{(p,q+1)} - \gamma^2 L_{q+1}^\top \left[L_{q+1}-L_q\right].
		\label{eq:inner_identities}
	\end{align}	%
	We now rewrite \eqref{eq:riccati_inner_loop_iter} in light of \eqref{eq:inner_identities} as %
	\begin{align}
		&A_{K,L}^{(p,q+1)^\top} P_{K,L}^{p,q} + P_{K,L}^{p,q}A_{K,L}^{(p,q+1)} - \gamma^2 \left[L_{q+1} - L_q\right]^\top L_{q+1} \nonumber \\
		& \,\, + Q_K -  \gamma^2 L_{q+1}^\top \left[L_{q+1} - L_q\right]-  \gamma^2 (L_{q}^\top L_{q}) = 0.
		\label{eq:riccati_inner_q}
	\end{align}
	At the $(q+1)$'st iteration, we have \eqref{eq:riccati_inner_loop_iter} as 
	\begin{align}
		&A_{K,L}^{(p,q+1)^\top} P_{K,L}^{p,q+1} +  P_{K,L}^{p,q+1}A_{K,L}^{(p,q+1)} + Q_K \nonumber \\
		& \qquad \qquad - \gamma^2 L_{q+1}^\top(K_p) L_{q+1}(K_p)  = 0.
		\label{eq:riccati_inner_qp1}
	\end{align}
	Subtracting \eqref{eq:riccati_inner_q} from \eqref{eq:riccati_inner_qp1}, we have 
	\begin{align}
		&A_{K,L}^{(p,q+1)^\top} \left[P_{K,L}^{p,q+1} - P_{K,L}^{p,q}\right] + \left[P_{K,L}^{p,q+1} - P_{K,L}^{p,q}\right] A_{K,L}^{(p,q+1)} + \nonumber \\
		& \qquad \qquad + \gamma^2 \left[L_{q+1}  - L_q\right]^\top \left[L_{q+1}  - L_q\right] = 0.
		\label{eq:riccati_inner_diff}
	\end{align}	 
	Since $\left[L_{q+1}  - L_q\right]^\top \left[L_{q+1}  - L_q\right] \succeq 0$, \eqref{eq:riccati_inner_diff} is indeed a Lyapunov equation so that $P_{K,L}^{p,q+1} \succeq P_{K,L}^{p,q}$ holds following Lemma \ref{lm:ZhouRobust}. Whence, we must have  $A_{K,L}^{(p,q+1)\top}$ Hurwitz. Following the argument for all $(q, {q}^\prime) \in \bar{q}$ with $q \neq q^\prime$, statement 2) holds.
	
	Observe: $P_{K,L}^{p,q}$ is self-adjoint by reason of \eqref{eq:riccati_outer_loop_iter}. By the theorem on the ``limit of monotonically decreasing operators"~\cite[pp. 190]{Kantarovich},  statement 2) implies that the sequence $\{ P_{K,L}^{p,\bar{q}}, \cdots, P_{K,L}^{p,q=0}  \}$ is monotonically decreasing and bounded from above by $P_{K,L}^{p,\bar{q}} \equiv P_{K,L}^{\star}$. That is, $P_{K,L}^{p,q}$ exists and is the solution of \eqref{eq:riccati_outer_loop_iter} and $P_{K,L}^{p,\bar{q}}$ is the unique positive definite solution to \eqref{eq:riccati_inner_loop_iter}. \textit{A fortiori}, we must have $\lim_{q \rightarrow \infty} P_{K,L}^{p,q}=P_{K,L}^{p,\infty}$. This establishes the third statement.
\end{proof}
\begin{proof}[Proof of Lemma \ref{lm:traceInner}]
	Subtracting \eqref{eq:lyapunovPKL} from \eqref{eq:brlemma2}, and using $L(K^\star) = \gamma^{-2} D^\top P_K$, we find that
	\begin{align}
		&(A_K+DL(K^\star))^\top(P_K - P_K^L) + (P_K - P_K^L)(A_K+  \\
		&\quad DL(K^\star))+ \Psi_K^L - \gamma^{-2}(P_K^L - P_K)DD^\top(P_K^L - P_K)= 0. \nonumber
	\end{align}
	Since $A_K + DL(K^\star)$ is Hurwitz, it follows from statements (1) and (3) of  Lemma~\ref{lm:ZhouRobust} 
	that we must have
	\begin{align}\label{eq:PK-PKLBound}
		P_K - P_K^L \preceq \int_{0}^{\infty} e^{(A_K + DL(K^\star))^\top t} \Psi_K^L e^{(A_K + DL(K^\star))t} \mathrm{d}t.
	\end{align}
	Taking the trace of the lhs, using~\cite[Theorem 2]{Mori1988}, and employing the cyclic property of the trace, the proof follows. 
\end{proof}
\begin{proof}[Proof of Lemma \ref{lm:robust_after_perturb}]
	\text{Let $F(\tilde{P} , \tilde{K})$ be }	
	\begin{align}
		&(A_K+\gamma^{-2}DD^\top P_K)^\top \tilde{P} + \tilde{P} (A_K+\gamma^{-2}DD^\top P_K) \nonumber \\
		& - \tilde{K}^\top B^\top (P_K + \tilde{P} ) - (P_K + \tilde{P} )B \tilde{K} + \tilde{K}^\top R K + K^\top R \tilde{K} \nonumber \\
		&+ \tilde{K}^\top R \tilde{K} + \gamma^{-2}\tilde{P}  DD^\top \tilde{P}.
	\end{align}
	\textbf{Observe}: the pair $(P_K+\tilde{P} , K+\tilde{K})$ satisfies \eqref{eq:brlemma2} iff $F(\tilde{P}, \tilde{K})=0$ and that $F(\tilde{P} , \tilde{K})=0$ implies an implicit function of $\tilde{P} $ with respect to $\tilde{K}$ since if $\tilde{P}  \in \bb{S}^n$ exists, $\tilde{K}$ must exist under the controllability and observability assumptions of Ass \ref{ass:realizability}. Let $\mathcal{F}(\tilde{P} , \tilde{K}) = \vec(F(\tilde{P} , \tilde{K}))$ so that,
	\begin{align}\label{eq:partialF}
		\begin{split}
			&\mathcal{F}(\tilde{P} , \tilde{K}) =\left[I_n \otimes (A_K+\gamma^{-2}DD^\top P_K)^\top  \right.\\
			& + \left.(A_K+\gamma^{-2}DD^\top P_K)^\top \otimes I_n\right] \vec(\tilde{P}) \\
			& - (P_K B \otimes I_n) \vec(\tilde{K}^\top) - (I_n \otimes P_KB)\vec(\tilde{K}) \\
			& - (I_n \otimes \tilde{K}^\top B^\top + \tilde{K}^\top B^\top \otimes I_n) \vec(\tilde{P} ) \\
			& + (K^\top R \otimes I_n)\vec(\tilde{K}^\top) + (I_n \otimes K^\top R) \vec(\tilde{K}) \\
			& + \vec(\tilde{K}^\top R \tilde{K}) + \gamma^{-2}\vec(\tilde{P} DD^\top \tilde{P} ).
		\end{split}
	\end{align}
	%
	Thus,
	\begin{align}
		&\frac{\partial \mathcal{F}(\tilde{P} , \tilde{K})}{\partial \vec(\tilde{P} )} = I_n \otimes [(A_K + \gamma^{-2}DD^\top P_K) - B\tilde{K}]^\top  \\
		&+ [(A_K + \gamma^{-2}DD^\top P_K) - B \tilde{K}]^\top \otimes I_n + \tilde{P}  DD^\top \otimes I_n \nonumber\\
		& + I_n \otimes \tilde{P}  DD^\top - (P_K B \otimes I_n) \vec(\tilde{K}^\top) - (I_n \otimes P_KB)  \nonumber \\
		& \vec(\tilde{K}) + (K^\top R \otimes I_n)\vec(\tilde{K}^\top) + (I_n \otimes K^\top R) \vec(\tilde{K}),  \nonumber 
	\end{align}
	where we have used~\cite[Theorem 9]{MAGNUS1985}, to obtain ${\partial \vec(\tilde{P}  DD^\top \tilde{P} )}/{\partial \vec(\tilde{P} )} = \tilde{P} DD^\top \otimes I_n + I_n \otimes \tilde{P} DD^\top$.
	%
	Since  $\mathcal{F}(0,0) = 0$, $(A_K + \gamma^{-2}DD^\top P_K)$ is Hurwitz, hence ${\partial \mathcal{F}(\tilde{P} , \tilde{K})}/{\partial \vec(\tilde{P})}\lvert_{\tilde{P} =0, \tilde{K} =0}$ is invertible. From the implicit function theorem,  there must exist an $e_1(K) >0$, such that $\tilde{P}$ is continuously differentiable with respect to $\tilde{K}$ for any $\tilde{K}\in \mathcal{B}_{e_1(K)}(0)$. Thus, $\norm{\tilde{P}} \to 0$ as $\norm{\tilde{K}} \to 0$. Since $K \in \mathcal{K}$ by ~\cite[Lemma A.1]{Zhang2021}, we must have $P_K \succ 0$. Therefore, there exists $e(K)>0$, such that $\sigma_{\max}(\tilde{P}) <  \sigma_{\min}(P_K)$, i.e. $P_K - \tilde{P} \succ 0$, as long as $\norm{\tilde{K}} < e(K)$. 
	
	Since $\tilde{P}$ and $\tilde{K}$ satisfy $F(\tilde{P}, \tilde{K}) = 0$, we have
	\begin{align}
		&(A-BK-B\tilde{K})^T (P_K + \tilde{P}) + (P_K + \tilde{P}) + Q + \\
		& (K + \tilde{K})^T R (K + \tilde{K}) + \gamma^{-2}(P_K + \tilde{P}) DD^T (P_K + \tilde{P}) =0\nonumber 
	\end{align}
	Since $P_K + \tilde{P} \succ 0$ when $\norm{\tilde{K}} < e(K)$, by~\cite[Lemma A.1]{Zhang2021}, $K + \tilde{K} \in \mathcal{K}$. That is, if $\tilde{K}$ is small, if we start the PI with a robustly stabilizing $K \in \mathcal{K}$, we can guarantee the feasibility of the iterates. 
\end{proof}

\begin{proof}[Proof of Lemma \ref{lm:ISS_InnerLoop}]
	Define $\tilde{L}_q^\top(K_p) = {L}_q^\top(K_p) - \hat{L}_q^\top(K_p)$ and $\tilde{P}_{K,L}^{p,q}={P}_{K,L}^{p,q}-\hat{P}_{K,L}^{p,q}$. Further, assume $\hat{A}_{K,L}^{p,q}$ is Hurwitz. From \eqref{eq:riccati_inner_loop_iter}, 
	\begin{align}
			&\hat{A}_{K,L}^{p,q\top} P_{K,L}^{p,q} + P_{K,L}^{p,q} \hat{A}_{K,L}^{p,q+1} + Q_K -\gamma^{-2}\hat{P}_{K,L}^{p,q} DD^\top \hat{P}_{K,L}^{p,q}  +  \nonumber \\
			&\gamma^{-2}(P_{K,L}^{p,q} - \hat{P}_{K,L}^{p,q})DD^\top (P_{K,L}^{p,q} - \hat{P}_{K,L}^{p,q}) - \tilde{L}_q^\top(K_p) D^\top P_{K,L}^{p,q} \nonumber  \\
			&\qquad - P_{K,L}^{p,q} D \tilde{L}_q(K_p) = 0.
	\end{align}
	Set $\norm{\tilde{L}_q^\top(K_p)} < {\sigma_{\min} (Q_K - \gamma^{-2}P_{K,L}^{p,q} D D^\top P_{K,L}^{p,q})}/\allowbreak{2\norm{D^\top P_{K,L}^{p,q}}}\triangleq e$. It follows from \eqref{eq:riccati_inner_qp1} that $Q \succeq \gamma^2 L_{q+1}^\top(K_p) L_{q+1}(K_p)$ by reason of it being admissible as a Lyapunov equation. The inequality $P_{K,L}^{p,q} \succeq \hat{P}_{K,L}^{p,q}$ so that
		\begin{align*}
			\begin{split}
				&-\gamma^{-2}\hat{P}_{K,L}^{p,q} DD^\top  \hat{P}_{K,L}^{p,q}  + \gamma^{-2}(P_{K,L}^{p,q} - \hat{P}_{K,L}^{p,q})DD^\top (P_{K,L}^{p,q} \\
				&\quad- \hat{P}_{K,L}^{p,q})  - (\tilde{L}_K^j)^\top D^\top P_{K,L}^{p,q} - P_{K,L}^{p,q} D \tilde{L}_q(K_p) + Q_K \succeq 0.
			\end{split}  
		\end{align*}
		Consequently, $\hat{A}_{K,L}^{p,q+1}$ is Hurwitz. Since $\hat{L}_q(K_0) = 0$ and $K \in \breve{\mathcal{K}}$, $\hat{A}_{K,L}^{p,0} = A-BK$ is Hurwitz. Hence, $\hat{A}_{K,L}^{p,q}$ is Hurwitz for all $q \in \mathbb{N}_+$ as long as $\norm{\tilde{L}_q(K_p)}_F \leq e$.
\end{proof}
	
\begin{proof}[Proof of Lemma \ref{lm:invariantOut}]
	Since $\mathcal{K}_h$ is compact, it follows from Lemma \ref{lm:robust_after_perturb} that $\underline{e}(h) :=\inf_{K \in \mathcal{K}_h}e(K)>0$. In addition, $\hat{K}' \in \mathcal{K}$ when $\norm{\tilde{K}} < \underline{e}(h)$. By~\cite[Lemma A.1]{Zhang2021}, ${P}_{\hat{K}'} = {P}_{\hat{K}'}^\top \succ 0$ is the solution of 
	\begin{align}\label{eq:AK'PK'}
		{A}_{\hat{K}'}^\top {P}_{\hat{K}'} + {P}_{\hat{K}'}{A}_{\hat{K}'} + {Q}_{\hat{K}'} + \gamma^{-2}{P}_{\hat{K}'}DD^\top {P}_{\hat{K}'} = 0,
	\end{align}
	where ${A}_{\hat{K}'} = A - B\hat{K}'$ and ${Q}_{\hat{K}'} = Q + (\hat{K}')^\top  R\hat{K}'$. Let ${A}_{\hat{K}'}^\star = A - B\hat{K}' + \gamma^{-2}DD^T P_{\hat{K}'}$.
	%
	%
	It follows from~\cite[Lemma A.1]{Zhang2021} that ${A}_{\hat{K}'}^\star$ is Hurwitz. 	Subtracting \eqref{eq:AK'PK'} from \eqref{eq:brlemma2}, using ${K}' = R^{-1}B^T P_K$, and completing the squares, 	%
	\begin{align}
		\begin{split}
			&{A}_{\hat{K}'}^{\star\top}(P_K - {P}_{\hat{K}'}) + (P_K - {P}_{\hat{K}'}){A}_{\hat{K}'}^\star \\
			&+ ({K}' - K)^\top R  ({K}' - K) - \tilde{K}^\top R \tilde{K}\\
			&+ \gamma^{-2}(P_K - {P}_{\hat{K}'})DD^\top (P_K - {P}_{\hat{K}'}) = 0.
		\end{split}
	\end{align}
	From Lemma \ref{lm:ZhouRobust}, we have $(P_K - {P}_{\hat{K}'}) \succeq$
	\begin{align}\label{eq:PKPK'Diff}
		\begin{split}
			& \int_{0}^{\infty} e^{A_{\hat{K}'}^{\star\top}t} E_K e^{A_{\hat{K}'}^{\star}t} \df t - \int_{0}^{\infty} e^{A_{\hat{K}'}^{\star\top}t}  \tilde{K}^T R \tilde{K} e^{A_{\hat{K}'}^{\star}t} \df t,
		\end{split}
	\end{align}
	so that taking the trace, using Lemma \ref{lm:mb_upper_bd} and~\cite[Theorem 2]{Mori1988},
	\begin{align}
		\begin{split}
			&\Tr(P_K - {P}_{\hat{K}'}) \geq \frac{\log(5/4)}{2\norm{{A}_{\hat{K}'}^\star}}\norm{E_K} \\
			&- \Tr\left( \int_{0}^{\infty} e^{A_{\hat{K}'}^{\star\top}t}  e^{A_{\hat{K}'}^{\star}t} \df t \right) \norm{R}\norm{\tilde{K}}^2.
		\end{split}
	\end{align}
	It follows from Lemmas \ref{lm:bound_EK} and \ref{lm:normTrace} that 
	\begin{align}\label{eq:PKPK'DiffTrace}
			&\Tr({P}_{\hat{K}'} - P^\star) \leq \left(1 -  
				\frac{\log(5/4) b(h)}{2n\norm{{A}_{\hat{K}'}^\star}}
			\right) \Tr({P}_{{K}} - P^\star) \nonumber \\
			&+ 
				\Tr\left( \int_{0}^{\infty} e^{A_{\hat{K}'}^{\star\top}t}  e^{A_{\hat{K}'}^{\star}t} \df t \right)
				 \norm{R}\norm{\tilde{K}}^2.
		\end{align}
			Let
			\[
			f_1(\hat{K}') = \frac{\log(5/4) b(h)}{2n\norm{{A}_{\hat{K}'}^\star}}, f_2(\hat{K}') = \Tr\left( \int_{0}^{\infty} e^{A_{\hat{K}'}^{\star\top t}}   e^{A_{\hat{K}'}^{\star t}} \df t \right).
			\]
			
		Since $f_1(\hat{K}')$ and $f_2(\hat{K}')$ are continuous with respect to $\hat{K}'$, 
		\begin{align}
			&\underline{f}_1(h)= \inf_{\hat{K}' \in  \mathcal{K}_{h}}{f}_1(\hat{K}') > 0,
			\bar{f}_2(h)= \sup_{\hat{K}' \in  \mathcal{K}_{h}}{f}_2(\hat{K}') < \infty.
		\end{align}
		It follows from \eqref{eq:PKPK'DiffTrace} that if $\norm{\tilde{K}} < \sqrt{\frac{\underline{f}_1(h) h}{\bar{f}_2(h)\norm{R}} }$, then $\Tr({P}_{\hat{K}'} - P^\star) < h$. In summary, if 
		\begin{align}
			\norm{\tilde{K}} < \min\bigg\{ \underline{e}(h), \sqrt{\frac{\underline{f}_1(h) h}{\bar{f}_2(h)\norm{R}} } \bigg\} =: f(h),
		\end{align}
		we have $\hat{K}' \in \mathcal{K}_h$.
\end{proof}

\begin{lemma}{Norm of a Matrix Trace~\cite[Theorem 4.2.2]{book_horn}}
	For any positive semi-definite matrix $P \in \mathbb{S}^n$, $\norm{P}_F \leq \Tr(P) \leq \sqrt{n}\norm{P}_F$, and $\norm{P} \leq \Tr(P) \leq n \norm{P}$. For any $x \in \mathbb{R}^n$, $x^\top P x \geq \eigmin(P) \norm{x}^{2}$.
	\label{lm:normTrace}
\end{lemma}
	
\begin{lemma} \label{lm:ZhouRobust}
	Assume $A \in \mathbb{R}^{n\times n}$ is Hurwitz and satisfies $A^\top P + PA +Q = 0$. Then, the following properties hold
	\begin{enumerate}[(1)]
		\item $P = \int_{0}^{\infty} e^{A^\top t}Qe^{At} \df t$;
		\item $P \succ 0$ if $Q\succ 0$, and $P \succeq 0$ if $Q\succeq 0$;
		\item If $Q\succeq 0$, then $(Q,A)$ is observable iff $P \succ 0$;
		\item For a $P^\prime \in \mathbb{S}^n$ satisfying $A^\top P^\prime + P^\prime A +Q^\prime = 0$, where $Q^\prime \preceq Q$, we have $P^\prime \preceq P$. 
	\end{enumerate}
\end{lemma}
\begin{proof}[Proof of Lemma \ref{lm:ZhouRobust}]
	The first three statements are proven in \cite[Lemma 3.18]{Zhou_Robust}. From the expression in the fourth statement above, $P^\prime$ can be expressed as
	\begin{align}
		&P^\prime = \int_{0}^\infty e^{A^\top t}Q^\prime e^{At} \df t
		\label{eq:zhou_pprime}
	\end{align}
	owing to statement (1) above. If $Q^\prime \preceq Q$, then comparing statement (1) with \eqref{eq:zhou_pprime}, we must have  $P^\prime \preceq P$.
\end{proof}

\begin{lemma}{\cite[Lemma 3.19]{Zhou_Robust}}
	Suppose that $P$ satisfies $A^\top P + PA +Q = 0$, then the following statements hold:
	\begin{enumerate}
		\item $A$ is Hurwitz if $P \succ 0$ and $Q \succ 0$.
		\item $A$ is Hurwitz if $P \succeq 0$, $Q \succeq 0$ and $(Q,A)$ is detectable.
	\end{enumerate}  
	\label{lm:inverseLya}
\end{lemma}

\begin{proof}[Proof of Theorem \ref{thm:innerISS}]
	When $\norm{\tilde{L}_q(K_p)} < e$, we have an Hurwitz $\hat{A}_{K,L}^{p,q}$ going by Lemma \ref{lm:ISS_InnerLoop}. Rewriting \eqref{eq:inexact_iter_iloop} for the $(p+1)$'th iteration and subtracting it from \eqref{eq:inexact_iter_iloop}, we have 
	\begin{align}
			&\hat{A}_{K,L}^{(p, q+1)\top}(\hat{P}_{K,L}^{(p, q+1)} - \hat{P}_{K,L}^{p, q} )+ (\hat{P}_{K,L}^{(p, q+1)} - \hat{P}_{K,L}^{p, q} )\hat{A}_{K,L}^{(p, q+1)\top} \nonumber \\
			&\,\, + \gamma^{-2}(\gamma^{2}\hat{L}_q(K_p) - D^\top \hat{P}_{K,L}^{p, q})^\top (\gamma^{2}\hat{L}_q(K_p) - D^\top \hat{P}_{K,L}^{p, q}) \nonumber \\
			&\qquad - \gamma^{2}\tilde{L}^\top_q(K_p) \tilde{L}_q(K_p)= 0.      
	\end{align}
	Suppose that $\hat{\Psi}_{K,L}^{p,q}= \gamma^{-2}(\gamma^{2}\hat{L}_q(K_p) - D^\top\hat{P}_{K,L}^{p,q})^\top(\gamma^{2}\hat{L}_q(K_p) - D^\top \hat{P}_{K}^j)$. It follows that since $\hat{A}_{K,L}^{p, q+1}$ is Hurwitz,  $\hat{P}_{K,L}^{(p, q+1)} - \hat{P}_{K,L}^{p, q} $ becomes
	\begin{align}
		\int_{0}^{\infty} e^{\hat{A}_{K,L}^{(p, q+1)^\top}t} \left[\hat{\Psi}_{K,L}^{p,q} - \gamma^{2} \tilde{L}_q^\top(K_p) \tilde{L}_q^\top(K_p) \right]e^{\hat{A}_{K,L}^{(p, q+1)}t} dt.   
	\end{align}
	%
	Now let $\hat{F}_K^q = \int_{0}^{\infty} e^{\hat{A}_{K,L}^{(p, q+1)^\top}t} \hat{\Psi}_{K,L}^{p,q} e^{\hat{A}_{K,L}^{(p, q+1)}t} dt$ so that
	\begin{align}
		\begin{split}
			&{P}_{K,L}^{p,q+1} - \hat{P}_{K,L}^{p,q+1} = {P}_{K,L} - \hat{P}_{K,L}^{p} - \hat{F}_K^q \\
			&+ \int_{0}^{\infty}e^{\hat{A}_{K,L}^{(p,q+1)^\top}t} \left( \gamma^{2}\tilde{L}_q^\top(K_p) \tilde{L}_q(K_p) \right) e^{\hat{A}_{K,L}^{p,q+1}t} dt.    
		\end{split}
		\label{eq:diffInner}
	\end{align}
	Let $f_K = \sup_{q \in \mathbb{N}_+}\norm{\hat{A}_{K,L}^{p,q+1}} $. From Lemma \ref{lm:traceInner}, we can write  $-\norm{\hat{F}_K^q} \le - \frac{\log(5/4)}{2f_K} \norm{ \hat{\Psi}_{K,L}^{p,q}}$. Furthermore, by Lemma \ref{lm:traceInner}, we can write $-\norm{ \hat{\Psi}_{K,L}^{p,q}} \leq - \frac{1}{c(K)}\Tr(P_{K,L}^{p,q} - \hat{P}_{K,L}^{p,q})$, where $c(K) = \Tr(\int_{0}^{\infty} e^{(A_K + DL_q(K_p))t} e^{(A_K + DL_q(K_p))^\top t} dt)$. Therefore, the trace of \eqref{eq:diffInner} becomes
	\begin{align}
			&\Tr({P}_{K,L}^{p,q} - \hat{P}_{K,L}^{p,q+1}) \leq \left(1-\frac{\log(5/4)}{2f_K c(K)}\right) \Tr({P}_{K} - \hat{P}_{K,L}^{p,q}) \nonumber \\
			&+ \Tr\left(\int_{0}^{\infty}e^{(\hat{A}_{K,L}^{p,q+1})t} e^{(\hat{A}_{K,L}^{p,q+1})^\top t} dt\right)\gamma^2 \norm{\tilde{L}_q(K_p)}^2.   
		\end{align}
		\begin{align}
			\text{Let } g = \sup_{q \in \mathbb{N}_+}\Tr\left(\int_{0}^{\infty}e^{(\hat{A}_{K,L}^{p,q+1})t} e^{(\hat{A}_{K,L}^{p,q+1})^\top t} dt\right),
		\end{align}
		and $
		\hat{\beta}(K) = 1-\frac{\log(5/4)}{2f_K c(K)}$, so that 
		\begin{align}
			\Tr({P}_{K,L}^{p,q} - \hat{P}_{K,L}^{p,q}) \leq \hat{\beta}^{q-1}(K)\Tr({P}_{K,L}^{p,q}) + \lambda(\norm{\tilde{L}}_\infty),
		\end{align}
		where $\lambda(\norm{\tilde{L}}_\infty) := \frac{1}{1-\hat{\beta}(K)}\gamma^2 g \norm{\tilde{L}}_\infty^2$. As $\norm{{P}_{K} - \hat{P}_{K,L}^{p,q}}_F \leq \Tr({P}_{K,L}^{p,q} - \hat{P}_{K,L}^{p,q})$, we establish the theorem.
\end{proof}
\end{document}